\documentclass[12pt]{article}
\usepackage{proof,prooftree}
\usepackage{amsfonts,amssymb}
\usepackage{amsthm}
\usepackage{latexsym}
\usepackage{pslatex}
\usepackage{stmaryrd}
\usepackage{fancybox}
\usepackage{amsmath} %
\usepackage{amssymb,prooftree,proof}
\usepackage{latexsym}
\usepackage[all]{xy}
\usepackage{qsymbols}
\usepackage{pslatex}
\usepackage{graphicx}
\usepackage{prooftree}
\usepackage{authblk}

\newcount\timehh\newcount\timemm \timehh=\time
\divide\timehh by 60 \timemm=\time
\newcommand{\timestamp}{
{\protect\small\sl\today\ --
  \ifnum\timehh<10 0\fi\number\timehh\,:\,
  \ifnum\timemm<10 0\fi\number\timemm}}

\newif\ifcomment

\usepackage[usenames,dvipsnames]{color}     

\definecolor{darkbrown}{cmyk}{.3,.75,.75,.15}

\newcommand{\via}[1]{}
\newcommand{\SKIP}[1]{}

\newcommand{\lG}{\ensuremath{\lambda^{\mathsf{Gtz}}}}

\newcommand{\Rcl}{\Lambda_{\circledR}}
\newcommand{\rcl}{\lambda_{\circledR}}

\newcommand{\llxr}{\lambda \mathsf{lxr}}

\newcommand{\rclsub}{\lambda_{\circledR}^{`[/]}}
\newcommand{\Rclsub}{\Lambda_{\circledR}^{`[/]}}
\newcommand{\Fvsub}{Fv^{`[/]}}
\newcommand{\vdashsub}{\vdash^{`[/]}}
\newcommand{\rclsubred}{"-{^{`[/]}}!!>"}
\newcommand{\rclsubredc}{"-{^{`[/]}}!!>!!>"}

\newcommand{\rightarrowc}{\rightarrow\!\!\!\!\!\!\!\rightarrow}

\newcommand{\dblsl}{/\hspace*{-3.5pt}/\hspace*{-6.3pt}/} 
\newcommand{\ISUB}[2]{|\hspace*{-3pt}|\hspace*{-3pt}|\hspace*{-3pt}[#1{\dblsl}#2]\hspace*{-3pt}|\hspace*{-3pt}|\hspace*{-3pt}|} 
\newcommand{\ISUBM}[4]{|\hspace*{-3pt}|\hspace*{-3pt}|\hspace*{-3pt}[#1{\dblsl}#2, \ldots, #3{\dblsl}#4]\hspace*{-3pt}|\hspace*{-3pt}|\hspace*{-3pt}|} 
\newcommand{\subnf}{\downarrow^{`[/]}}
\newcommand{\multsbt}{\mathcal{M}\textsl{ul}}

\newcommand{\isub}[2]{[#1/#2]}   

\newcommand{\weak}[2]{#1 \odot #2}
\newcommand{\cont}[4]{#1 < ^{#2}_{#3}#4}

\newcommand{\bnorm}[1]{|\!|#1|\!|_{`[/]}}
\newcommand{\cntx}{\mathcal{C}}

\newcommand{\isubs}[2]{#1/#2}   
\newcommand{\tlam}{{Types}} 
\newcommand{\lefti}{[ \! [}      
\newcommand{\righti}{] \! ]}
\newcommand{\ti}[1]{\lefti #1 \righti}
\newcommand{\tei}{\ti} 
\newcommand{\NF}{\ensuremath{\mathcal{NF}}} 
\newcommand{\SN}{\ensuremath{\mathcal{SN}}} 

\newcommand{\VAR}{\textsf{VAR}}
\newcommand{\RED}{\textsf{RED}}
\newcommand{\SAT}{\textsf{EXP}}  
\newcommand{\WEAK}{\textsf{THIN}}
\newcommand{\CONT}{\textsf{CONT}}

\newcommand{\vX}{\mathcal{X}} 
\newcommand{\vM}{\mathcal{M}} 
\newcommand{\vN}{\mathcal{N}} 
\newcommand{\tA}{\alpha} 
\newcommand{\tB}{\beta} 
\newcommand{\tC}{\gamma} 
\newcommand{\tS}{\sigma} 
\newcommand{\tT}{\tau} 
\newcommand{\tR}{\rho} 
\newcommand{\tU}{\upsilon}

\newcommand{\fsto}{\xymatrix@C=15pt{\ar@{>}[r] &}}

\newcommand{\LR}{\Lambda_{\circledR}}

\newcommand{\dztop}{\Delta_{0}^{\top}}
\newcommand{\dztopp}{{\Delta^{'}_{0}}^{\top}}
\newcommand{\dztoppp}{{\Delta^{''}_{0}}^{\top}}

\newcommand{\gztopone}{\Gamma_{0}^{1\top}}
\newcommand{\gztopn}{\Gamma_{0}^{n\top}}
\newcommand{\gtop}{\Gamma^{\top}}
\newcommand{\rsc}[1]{[ #1 ]_{rc}}
\newcommand{\IntR}[1]{[ #1 ]_{\circledR}}

\theoremstyle{plain}
\newtheorem{theorem}{Theorem}
\newtheorem{proposition}[theorem]{Proposition}
\newtheorem{lemma}[theorem]{Lemma}

\newtheorem{example}[theorem]{Example}
\newtheorem{definition}[theorem]{Definition}

 \newif\ifcomment
 \commentfalse           
 
\begin{document}


 \title{Resource control and intersection types:\\ an intrinsic connection}

 \author[1]{S. Ghilezan} 
 \author[1]{J.~Iveti\' c}  
 \author[2]{P. Lescanne} 
 \author[3]{S. Likavec} 

 \affil[1]{University of Novi Sad, Faculty of Technical Sciences,  Serbia}
 \affil[2]{University of Lyon, \' Ecole Normal Sup\' erieure de Lyon, France}
 \affil[3]{Dipartimento di Informatica, Universit\`a di Torino, Italy}

  \date{\today}

\maketitle

 \begin{abstract}
 In this paper we investigate the $\rcl$-calculus, a $\lambda$-calculus enriched with resource control.
 Explicit control of resources is enabled by the presence of erasure and duplication operators, which correspond to thinning and contraction rules in the type assignment system.
 We introduce directly the class of $\rcl$-terms and we provide
 a new treatment of substitution by its decomposition into atomic steps. 
 We propose an intersection type assignment system for $\rcl$-calculus which makes a clear  correspondence between three roles of variables and three kinds of intersection types. 
 Finally, we provide the characterisation of strong normalisation in $\rcl$-calculus by means of an intersection type assignment system. This process uses typeability of normal forms, redex subject expansion and reducibility method.
 \end{abstract}



 \section*{Introduction}
 \label{sec:intro}
 
The notion of resource awareness and control has gained an important role both in theoretical and applicative domains: in logic and lambda calculus 
 as well as in programming langugages and compiler design.
 The idea to control the use of formulae is present in Gentzen's structural rules (\cite{gent35}), whereas the idea to control the use of variables can be traced  back to Church's $\lambda I$-calculus (e.g. ~\cite{bare84}).
The augmented ability to control the number and order of uses of operations and objects has a wide range of applications which enables, among others, compiler optimisations and  memory management that prevents memory leaking (e.g. ~\cite{walk05}).

In this paper, we investigate the control of resources in the \mbox{$\lambda$-calculus}. We propose the $\rcl$-calculus, a \mbox{$\lambda$-calculus} enriched with resource control operators.
The explicit control of resources is enabled by the presence of \emph{erasure} and \emph{duplication} operators, which correspond to thinning and contraction rules in the type assignment system. Erasure is the operation that indicates that a variable is not present in the term anymore, whereas duplication indicates that a variable will have two occurrences in the term which receive specific names to preserve the ``linearity'' of the term. Indeed,
 in order to control all resources, in the spirit of the $\lambda
I$-calculus (see e.g.~\cite{bare84}), void lambda abstractions are
not acceptable, so in order to have $\lambda x.M$ well-formed the variable~$x$
has to occur in~$M$. But if $x$ is not used in the term $M$, one must
perform an \emph{erasure} by using the
expression $\weak{x}{M}$. In this way, the term $M$ does not
contain the variable~$x$, but the term $\weak{x}{M}$ does.
Similarly, a variable should not occur twice. If nevertheless, we
want to have two positions for the same variable, we have to
duplicate it explicitly, using fresh names. This is done by using
the operator $\cont{x}{x_1}{x_2}{M}$, called \emph{duplication} which creates two fresh variables $x_1$ and
$x_2$.

\paragraph{Outline of the paper}
We first introduce the syntax and reduction rules of the $\rcl$-calculus (Section~\ref{sec:syntax}).
We then introduce intersection types into the $\rcl$-calculus (Section~\ref{sec:types}). Finally, by means of intersection types, we completely caracterise strong normalisation in $\rcl$ (Section~\ref{sec:typeSN}).

\paragraph{Section~\ref{sec:syntax}}
We first introduce the syntax and reduction rules of the $\rcl$-calculus. Explicit control of erasure and duplication leads to decomposition of
reduction steps into more atomic steps, thus revealing the
details of computation which are usually left implicit. Since
erasing and duplicating of (sub)terms essentially changes the
structure of a program, it is important to see how this mechanism
really works and to be able to control this part of computation.
We chose a direct approach to term calculi rather than taking a
more common path through linear logic \cite{abra93,bent93}.

Although the design of our calculus  has been 
motivated by theoretical considerations, 
it may have practical implications as well. Indeed, for instance in the description of compilers by rules
with binders~\cite{rose:LIPIcs:2011:3130,rose11:_implem_trick_that_make_crsx_tick},
the implementation of substitutions of linear variables by
inlining\footnote{\emph{Inlining} is the technique which
consists in copying at compile time the text of a function instead of implementing a call to that function.} is simple and efficient when substitution of duplicated
variables requires the cumbersome and time consuming mechanism of
pointers and it is therefore important to tightly control
duplication. On the other hand, a precise control of erasing does
not require a garbage collector and prevents memory leaking.

\paragraph{Section~\ref{sec:types}}
Intersection types were introduced in~\cite{coppdeza78,coppdeza80,pott80,sall78} to overcome the limitations of the simple type discipline in which the only forming operator is an
arrow~$\rightarrow$.
The newly obtained intersection type assignment systems enable a
complete characterisation of termination of term
calculi~\cite{bake92,gall98,ghil96}. 
Later on, intersection types became a powerful tool for characterising strong normalisation in different calculi~\cite{dougghillesc07,kikuchiRTA07,matthes2000,neer05}.

We propose an intersection type assignment system $\rcl\cap$
that integrates intersection into logical rules, thus preserving syntax-directedness of the system. We assign a restricted form of intersection types to terms, namely strict types, therefore minimizing the need for pre-order on types.

Intersection types in the presence of resource control operators were firstly introduced in~\cite{ghilivetlikalesc11}, where two systems with idempotent intersection were proposed. Later, non-idempotent intersection types for contraction and weakening are treated in \cite{bernleng13}. In this paper, we treat a general form of intersection without any assumptions about idempotence. As a consequence, our intersection type system can be considered both as idempotent or as non-idempotent, both options having their benefits depending on the motivation.

Intersection types fit naturally with resource control. Indeed, the control allows us to consider three roles of variables: variables as placeholders (the traditional view of $`l$-calculus), variables to be duplicated and variables to be erased because they are irrelevant. For each kind of a variable, there is a kind of type associated to it, namely a strict type for a \emph{placeholder}, an intersection type for a variable \emph{to-be-duplicated}, and a specific type $\top$ for an \emph{erased} variable.

\paragraph{Section~\ref{sec:typeSN}}
By the means of the introduced intersection type assignment system $\rcl\cap$, we manage to completely characterise strong normalisation in  $\rcl$, i.e.\
we prove that
terms in the  $\rcl$-calculus enjoy strong normalisation if and only if they are typeable in $\rcl\cap$. First, we prove that all strongly normalising terms are typeable in the $\rcl$-calculus by using typeability of normal forms and redex subject expansion. We then prove that terms typeable in $\rcl$-calculus are strongly normalising by adapting the reducibility method for explicit resource control operators.



\paragraph{Main contributions}

The main contributions of this paper are:
\begin{itemize}
\item[(i)] an improved presentation of resource control lambda calculus syntax  with a direct definition of the syntax of resource control terms.  Other presentations define first an unconstrainted syntax of terms with duplication and erasure which is later restricted to linear terms;
\item[(ii)] a new treatment of substitution and its decomposition into more atomic steps;
\item[(iii)] an intersection type assignment system for resource control lambda calculus which makes explicit the intrinsic correspondence between three kinds of variables and three kinds of intersection types;
\item[(iv)] a characterisation of strong normalisation in $\rcl$-calculus by means of an intersection type assignment system, by using typeability of normal forms, redex subject expansion and reducibility.
\end{itemize}

\tableofcontents

\section{Resource control lambda calculus $\rcl$}
\label{sec:syntax}


The \emph{resource control} lambda calculus, $\rcl$, is an extension of the
$\lambda$-calculus with explicit erasure and duplication.

\subsection{Syntax}

Terms and lists, respectively sets, of free variables in $\rcl$ are mutually recursively defined.

\begin{definition}\label{def:rcl-terms}
\rule{0in}{0in}
\begin{enumerate}
\item[(i)] The set of $\rcl$-terms, denoted by $\LR$, is defined by inference rules given in Figure~\ref{fig:wf}.
\item[(ii)] The list of free variables of a term $M$, denoted by $Fv[M]$, is defined by inference rules given in Figure~\ref{fig:freevar}.
\item[(iii)] The set of free variables of a term $M$, denoted by $Fv(M)$, is obtained from the list $Fv[M]$  by unordering.
\end{enumerate}
\end{definition}

\begin{figure}[htpb]
\centerline{ \framebox{ $
    \begin{array}{c}
    \\
    \infer[(var)]{x \in \LR}{}
    \\\\
      \begin{array}{c@{\quad\quad}c}
        \infer[(abs)]{\lambda x.M \in \LR}
        {M \in \LR \;\; x \in Fv(M)} &
        \infer[(app)]{MN \in \LR}
        {M\in \LR\;\; N \in \LR \;\; Fv(M) \cap Fv(N) = \emptyset}
      \end{array}
      \\\\
      \begin{array}{c} 
      \infer[(era)]{\weak{x}{M} \in \LR}
      {M \in \LR \;\; x \notin Fv(M)} 
      \end{array}
      \\\\
      \begin{array}{c} 
      \infer[(dup)]{\cont{x}{x_1}{x_2}{M} \in \LR}
      {M \in \LR \;\;\;  x_1,x_2 \in Fv(M) \;\;\; x_{1} \neq x_{2}\;\;\; x \notin Fv(M) \setminus \{x_{1}, x_{2}\}}
    \end{array}\\ \\
  \end{array}
$ }} \caption{$\LR$: the set of $\rcl$-terms} \label{fig:wf}
\end{figure}

\begin{figure}[htpb]
\centerline{ \framebox{ $
    \begin{array}{c}
    \\
    \begin{array}{c@{\qquad\qquad}c}
        \infer[]{Fv[x]=[x]}{} &
        \infer[]{Fv[\lambda x_i.M] = [x_1,x_2,...x_{i-1},x_{i+1},...,x_m] }
        {Fv[M]=[x_1,x_2,...,x_m]}
    \end{array}
    \\ \\
    \begin{array}{c@{\quad\quad}c}
        \infer[]{Fv[MN] = [x_1,...,x_m,y_1,...,y_n]}
        {Fv[M] = [x_1,...,x_m] \;\;Fv[N] = [y_1,...,y_n]} &
        \infer[]{Fv[\weak{x}{M}] = [x, x_1,...,x_m] }
      {Fv[M] = [x_1,...,x_m]}
    \end{array}
    \\ \\
    \begin{array}{c}
    \infer[]{Fv[\cont{x}{x_i}{x_j}{M}] = [x,x_1,...x_{i-1},x_{i+1},......x_{j-1},x_{j+1},...,x_m]}
      {Fv[M] = [x_1,...,x_m]}
    \end{array}\\ \\
  \end{array}
$ }} \caption{List of free variables of a $\rcl$-term} \label{fig:freevar}
\end{figure}

A $\rcl$-term, ranged over by $M,N,P,...,M_1,...$, can be a variable from an enumerable set $\LR$ (ranged over by $x,y,z,x_1, \ldots$), 
an abstraction, an application, an erasure or a duplication.
The duplication $\cont{x}{x_1}{x_2}{M}$ binds the variables
$x_1$ and $x_2$ in $M$ and introduces a free variable $x$. The
erasure $\weak {x}{M}$ introduces also a free variable $x$.
In order to avoid parentheses, we let the scope of all binders extend
to the right as much as possible.

Informally, we say that a term is an expression in which every free variable occurs
exactly once, and every binder binds (exactly one occurrence of) a free variable. Our
notion of terms corresponds to the notion of linear terms in \cite{kesnleng07}. In
that sense, only linear expressions are in the focus of our investigation. In other
words, a term is well-formed in $\rcl$ if and only if bound variables appear actually
in the term and variables occur at most once. This assumption is not a restriction, since every 
pure $\lambda$-term has a
corresponding $\rcl$-term and vice versa, due to the embeddings  
given in Definition~\ref{def:adding-res} and ~\ref{def:removing-res}
and illustrated by Example~\ref{exa:adding-res}.

\begin{definition}\label{def:adding-res}
The mapping $\rsc{\;\;}:\Lambda\;\to\;\LR$ is
defined in the following way:
$$
\begin{array}{rcl}
\rsc{x} & = & x\\
\rsc{\lambda x.t} & = & \left\{
                        \begin{array}{ll}
                        \lambda x. \rsc{t}, & x \in Fv(t)\\
                        \lambda x. \weak{x}{\rsc{t}}, & x \notin Fv(t)
                        \end{array}
                        \right.\\
\rsc{MN} & = & \left\{
               \begin{array}{ll}
               \rsc{t}\rsc{s}, & Fv(t) \cap Fv(s) = \emptyset\\
               \cont{x}{x_1}{x_2}{\rsc{t\isub{x_1}{x}s\isub{x_2}{x}}}, & x \in Fv(t) \cap Fv(s)
               \end{array}
               \right.
\end{array}
$$
\end{definition}

Reciprocally, a $\rcl$-term has a corresponding $`l$-term.
  \begin{definition}\label{def:removing-res}
    The mapping $\IntR{\;\;}:\LR\;\to\;\Lambda$ is defined in the following way:
    \begin{eqnarray*}
      \IntR{x} &=& x\\
      \IntR{`l x.M}&=& `l x.\IntR{M}\\
      \IntR{M\,N}&=& \IntR{M}\,\IntR{N}\\
      \IntR{\cont{x}{x_1}{x_2}{M}} &=& \IntR{M}\isub{x}{x_1}\isub{x}{x_2}\\
      \IntR{\weak{x}{M}} &=&  \IntR{M}
    \end{eqnarray*}
  \end{definition}
  
\begin{proposition}
\rule{0in}{0in}
\begin{itemize}
\item[(i)] For each pure lambda term $t \in \Lambda$ there is a term $M \in \LR$ such that $\rsc{t} = M$.
\item[(ii)] For each resource lambda term $M \in \LR$ there is a term $t \in \Lambda$ such that $\IntR{M} = t$.
\end{itemize}
\end{proposition}

\begin{example}\label{exa:adding-res}
  Pure $\lambda$-terms $\lambda x.y$ and $\lambda x. xx$ are not $\rcl$-terms, whereas $\rsc{\lambda x.y}=\lambda x.(\weak{x}{y})$ and $\rsc{\lambda x. xx}=\lambda
  x.\cont{x}{x_1}{x_2}(x_1 x_2)$ are both $\rcl$-terms.

\medskip

  \begin{small}
    \begin{displaymath}
      \prooftree 
      \prooftree 
      \prooftree 
      \justifies y`:\Rcl \quad
      \using (var)
      \endprooftree 
      x \notin Fv(y)
      \justifies \weak{x}{y}`:\Rcl
      \using (era)
      \endprooftree 
      \quad
      x \in Fv(\weak{x}{y})
      \justifies `l x. \weak{x}{y}`:\Rcl
      \using (abs)
      \endprooftree 
    \end{displaymath}
    \begin{displaymath}
      \prooftree 
      \prooftree 
      \prooftree 
      \vdots
      \justifies x_1 x_2 `:\Rcl
      \endprooftree 
      \quad
      x \notin Fv(x_1 x_2)\setminus \{x_1,x_2\}
      \quad
      x_1,x_2 \in Fv(x_1 x_2)
      \justifies \cont{x}{x_1}{x_2}(x_1 x_2)`:\Rcl
      \using (dup)
      \endprooftree 
      \quad
      x \in Fv(\cont{x}{x_1}{x_2}(x_1 x_2))
      \justifies `l x.\cont{x}{x_1}{x_2}(x_1 x_2)`:\Rcl
      \using (abs)
      \endprooftree 
    \end{displaymath}
  \end{small}

\medskip

\end{example}
In the sequel, we use the following abbreviations:
\begin{enumerate}
\item[$\bullet$] $\weak{x_1}{...\;\weak{x_n}{M}}$ is abbreviated to $\weak{X}{M}$,
  when  $X$ is the list $[x_1,x_2,...,x_n]$;
\item[$\bullet$] $\cont{x_1}{y_1}{z_1}{...\;\cont{x_n}{y_n}{z_n}{M}}$ is abbreviated
  to $\cont{X}{Y}{Z}{M}$ if $X$ is the list $[x_1,x_2,...,x_n]$,\\
 $Y$ is the list $[y_1,y_2,...,y_n]$ and $Z$ is the list $[z_1,z_2,...,z_n]$.
\end{enumerate}

Notice that $X$, $Y$ and $Z$ are lists of equal length.
If $n=0$, i.e.\ if $X$, $Y$ and $Z$ are the empty lists, then $\weak{X}{M} = \cont{X}{Y}{Z}{M} = M$.
Note that later on due to the equivalence relation defined in Figure~\ref{fig:equiv-rcl}, in $\weak{X}{M}$
we can take $X$ to be the set $\{x_1,x_2,...,x_n\}$.

In what follows we use Barendregt's convention~\cite{bare84} for variables: in the
same context a variable cannot be both free and bound. This applies to binders like
$`l x.M$ which binds $x$ in $M$ and $\cont{x}{x_1}{x_2}M$ which binds $x_1$ and $x_2$ in
$M$.
\subsection{Substitution}
At this point, we chose to introduce a \emph{substitution operator} to define
  \emph{substitution} in $\Rcl$. Due to its interference with the linearity of
  terms and its slight difference with the standard substitution of the
  $`l$-calculus, the concept of substitution has to be carefully defined in
  the $\rcl$-calculus. For that reason, in Definition~\ref{def:rclsub-terms} we first make precise
  the syntax of $\rclsub$, i.e.\ the language $\rcl$ extended with a substitution operator, by
  providing mutually recursive definitions of $\rclsub$-terms and lists (respectively sets) of free variables (see Figures~\ref{fig:wfsub} and \ref{fig:freevarsub}).

\begin{definition}\label{def:rclsub-terms}
\rule{0in}{0in}
\begin{enumerate}
\item[(i)] The set of $\rclsub$-terms, denoted by $\Rclsub$, is defined by inference rules given in Figure~\ref{fig:wfsub}.
\item[(ii)] The list of free variables of a $\rclsub$-term $M$, denoted by $\Fvsub[M]$, is defined by inference rules given in Figure~\ref{fig:freevarsub}.
\item[(iii)] The set of free variables of a $\rclsub$-term $M$, denoted by $\Fvsub(M)$, is obtained from the list $\Fvsub[M]$  by unordering.
\end{enumerate}
\end{definition}

\begin{figure}[htpb]
\centerline{ \framebox{ $
    \begin{array}{c}
    \\
    \infer[(var)]{x \in \Rclsub}{}
    \\\\
      \begin{array}{c@{\quad\quad}c}
        \infer[(abs)]{\lambda x.M \in \Rclsub}
        {M \in \Rclsub \;\; x \in \Fvsub(M)} &
        \infer[(app)]{MN \in \Rclsub}
        {M\in \Rclsub\;\; N \in \Rclsub \;\; \Fvsub(M) \cap \Fvsub(N) = \emptyset}
      \end{array}
      \\\\
      \begin{array}{c} 
      \infer[(era)]{\weak{x}{M} \in \Rclsub}
      {M \in \Rclsub \;\; x \notin \Fvsub(M)} 
      \end{array}
      \\\\
      \begin{array}{c} 
      \infer[(dup)]{\cont{x}{x_1}{x_2}{M} \in \Rclsub}
      {M \in \Rclsub \;\;\; x \notin \Fvsub(M) \setminus \{x_{1}, x_{2}\} \;\;\; x_1,x_2 \in \Fvsub(M) \;\;\; x_{1} \neq x_{2}}
    \end{array}
    \\\\
    \begin{array}{c}
      \infer[(sub)]{M\isub{N}{x} \in \Rclsub}
        {M \in \Rclsub \quad x \in \Fvsub(M) \quad N \in \Rcl \quad \Fvsub(M)\setminus
          \{x\}\ \cap Fv(N) =\emptyset}
    \end{array}\\\\
  \end{array}
$ }} %
\caption{$\Rclsub$: the set of $\rclsub$-terms} \label{fig:wfsub}
\end{figure}

\begin{figure}[htpb]
\centerline{ \framebox{ $
    \begin{array}{c}
    \\
    \begin{array}{c@{\quad\quad}c}
        \infer[]{\Fvsub[x]=[x]}{} &
        \infer[]{\Fvsub[\lambda x_i.M] = [x_1,x_2,...x_{i-1},x_{i+1},...,x_m] }
        {\Fvsub[M]=[x_1,x_2,...,x_m]}
    \end{array}
    \\ \\
    \begin{array}{c@{\quad\quad}c}
        \infer[]{\Fvsub[MN] = [x_1,...,x_m,y_1,...,y_n]}
        {\Fvsub[M] = [x_1,...,x_m] \;\;\Fvsub[N] = [y_1,...,y_n]} &
        \infer[]{\Fvsub[\weak{x}{M}] = [x, x_1,...,x_m] }
      {\Fvsub[M] = [x_1,...,x_m]}
    \end{array}
    \\ \\
    \begin{array}{c}
    \infer[]{\Fvsub[\cont{x}{x_i}{x_j}{M}] = [x,x_1,...x_{i-1},x_{i+1},......x_{j-1},x_{j+1},...,x_m]}
      {\Fvsub[M] = [x_1,...,x_m]}
    \end{array}\\ \\
    \begin{array}{c}
      \infer[]{\Fvsub[M\isub{N}{x_i}] = [x_1,x_2,...x_{i-1},x_{i+1},...,x_m,y_1,...,y_n]}
      {\Fvsub[M] = [x_1,...,x_m] & Fv[N] = [y_1,...,y_n]}
    \end{array}\\\\
  \end{array}
$ }} %
\caption{List of free variables of a $\rclsub$-term} \label{fig:freevarsub}
\end{figure}

 Notice that the set $\Rcl$ is a strict subset of the set $\Rclsub$, $\Rcl \subset \Rclsub$, and that $N$ in
 $M \isub{N}{x}$ is substitution free, therefore we can write both $\Fvsub(N)$ and $Fv(N)$ for $N$ in $M \isub{N}{x}$. Also, notice that if a term $M$ is substitution free, then $\Fvsub(M)=Fv(M)$. Barendregt's convention applies to the substitution operator as well, where $M[N/x]$ can be seen as a binder for $x$ in $M$. \\ \\

\begin{definition}\label{def:rclsub-eval}\rule{0in}{0in}
\begin{itemize}
\item[(i)]
The evaluation of the substitution operator in the $\rclsub$-term $M \isub{N}{x}$, denoted by $\rclsubred$, is defined by the rules given in Figure~\ref{fig:sub-rcl}. As usual, it is closed under $\alpha$-equivalence and regular contexts.
In the last row in Figure~\ref{fig:sub-rcl}, terms $N_1$ and $N_2$ are obtained from the term $N$ by renaming of its free variables, i.e.\ by substitution of all free variables of $N$ by fresh variables, respectively.
\item[(ii)]
$\rclsubredc$ is the reflexive, transitive closure of $\rclsubred$.
\end{itemize}
\end{definition}




\begin{figure}[ht]
\centerline{ \framebox{ $
\begin{array}{rcl}
x\isub{N}{x} & \rclsubred & N \\
(\lambda y.M)\isub{N}{x} & \rclsubred  & \lambda y.M\isub{N}{x},\;\;x \neq y \\
(MP)\isub{N}{x} & \rclsubred  & M\isub{N}{x} P, \;\;x \in \Fvsub(M) \\
(MP)\isub{N}{x} & \rclsubred  & M P\isub{N}{x}, \;\;x \in \Fvsub(P) \\
(\weak{y}{M})\isub{N}{x} &\rclsubred  & \weak{y}{M\isub{N}{x}}, \;\;x \neq y\\
(\weak{x}{M})\isub{N}{x} & \rclsubred  & \weak{Fv(N)}{M}\\
(\cont{y}{y_1}{y_2}{M})\isub{N}{x} & \rclsubred  &
\cont{y}{y_1}{y_2}{M\isub{N}{x}}, \;\;x \neq y\\
(\cont{x}{x_1}{x_2}{M})\isub{N}{x} & \rclsubred  &
\cont{Fv[N]}{Fv[N_1]}{Fv[N_2]}{M\isub{N_1}{x_1}\isub{N_2}{x_2}}\\
\end{array}
$ }}  %
\caption{Evaluation of the substitution operator in the $\rclsub$-calculus}
\label{fig:sub-rcl}
\end{figure}

For a full understanding of the role of $\rclsub$, we would like to stress 
two facts:
  \begin{itemize}
  \item $\rclsubred$ is the operational definition of the substitution in $\Rcl$.
  \item  $\rclsubred$ is used with a higher priority than the reductions of $\rcl$
    given in Figure~\ref{fig:red-rcl} (because it is used to define substitution in $\Rcl$).
  \end{itemize}

To summarise, we have added a new operator to the syntax of $\rcl$ called
\emph{substitution operator} and denoted by $\isub{\;}{\;}$, and defined the evaluation of the substitution operator, which brings us to
$\rclsub$-calculus.

We prove the following safety property. 
\begin{proposition}\label{prop:subsubject}
\rule{0in}{0in}
  \begin{itemize}
  \item[(i)] If $Q \rclsubredc R$ and $Q\in\Rclsub$,  then  $R\in\Rclsub$.
  \item[(ii)] If $Q \rclsubredc R$ then  $\Fvsub(Q)=\Fvsub(R)$.
  \end{itemize}
\end{proposition}

\begin{proof}
  These properties are preserved by context. Therefore we can restrict our proof to the case
  where $Q$ is the instance of the left-hand side of a rule in
  Figure~\ref{fig:sub-rcl} and consider only one-step reduction $\rclsubred$.
  We consider only two paradigmatic rules.
  \begin{sloppypar}
    \begin{itemize}
    \item $(M\;P)\isub{N}{x} \ \rclsubred M\isub{N}{x}\;P$ with $x`:\Fvsub(M)$.
      \begin{itemize}
      \item We know that $x`:\Fvsub(M)$. Then $(M\;P)\isub{N}{x}`:\Rclsub$ means that
        $M`:\Rclsub$, $P`:\Rclsub$, $\Fvsub(M)\cap\Fvsub(P) {= \emptyset}$,
        $N`:\Rcl$ and ${(\Fvsub(M) \cup\Fvsub(P)) \setminus \{x\}} \cap Fv(N) =
        \emptyset$.  On the other hand, $M\isub{N}{x}\;P`:\Rclsub$ means
        $M`:\Rclsub$, $N`:\Rcl$, $P`:\Rclsub$ and
        $\Fvsub(M\isub{N}{x})\cap\Fvsub(P) = \emptyset$. Since
        $\Fvsub(M)\cap\Fvsub(P) {= \emptyset}$ and $((\Fvsub(M) \cup\Fvsub(P))
        \setminus \{x\}) \cap Fv(N) = \emptyset$, this implies
        $\Fvsub(M\isub{N}{x})\cap\Fvsub(P) = \emptyset$, hence the condition on free
        variables for $M\isub{N}{x}\;P$ is fulfilled.
      \item $\Fvsub((M\;P)\isub{N}{x}) = {\Fvsub(M\;P) \setminus \{x\} \  \cup Fv(N)} =\\
        {(\Fvsub(M) \cup \Fvsub(P)) \setminus \{x\} \ \cup Fv(N)} =\\ {(\Fvsub(M) \cup
          Fv(N)) \setminus \{x\} \ \cup \Fvsub(P)} = {\Fvsub(M\isub{N}{x}\;P)}$.
      \end{itemize}
    \item $(\weak{x}{M})\isub{N}{x} \rclsubred \weak{Fv(N)}{M}$.
      \begin{itemize}
      \item $(\weak{x}{M})\isub{N}{x}`:\Rclsub$ means $M`:\Rclsub$, $x`;\Fvsub(M)$,
        $N`:\Rcl$ and $\Fvsub(M) \cap Fv(N)= \emptyset$. On the other hand,
        $\weak{Fv(N)}{M}`:\Rclsub$ means $M`:\Rclsub$ and $Fv(N)\cap
        \Fvsub(M)=\emptyset$.
      \item $\Fvsub((\weak{x}{M})\isub{N}{x}) = \Fvsub(M)\cup Fv(N) =
        \bigcup_{y`:Fv(N)}\{y\} \cup \Fvsub(M) = \Fvsub(\weak{Fv(N)}{M})$.
      \end{itemize}
    \end{itemize}
  \end{sloppypar}

\end{proof}

Figure~\ref{fig:sub-rcl} defines the evaluation of substitution in $\Rcl$. Indeed,
the reduction $\rclsubredc$ terminates (Proposition~\ref{prop:box-term}) and when it
terminates it yields actually a term in $\Rcl$, i.e.\ there is no more substitution operator in the resulting term (Proposition~\ref{prop:nf}). Therefore, there is no need for defining evaluation of $M\isub{N}{x}$ in case of $M \equiv Q\isub{P}{y}$, because Propositions~\ref{prop:box-term} and \ref{prop:nf} guarantee that $Q\isub{P}{y}$ will be evaluated to some $Q' \in \Rcl$, thus $Q\isub{P}{y}\isub{N}{x} \rclsubredc Q' \isub{N}{x}\rclsubredc Q''$, for some $Q'' \in \Rcl$.

In order to prove normalisation in Proposition~\ref{prop:box-term}, we introduce the following measure.
\begin{definition}
\label{def-measure}
The measure $\bnorm{\cdot}$ on $\rclsub$-terms is defined as follows:
  \begin{eqnarray*}
    \bnorm{x} &=& 1\\
    \bnorm{`l x.M} &=& \bnorm{M} + 1\\
    \bnorm{M\;N} &=& \bnorm{M} + \bnorm{N} + 1 \\
    \bnorm{\weak{x}{M}} &=& \bnorm{M} + 1\\
    \bnorm{\cont{x}{y}{z}{M}}&=& \bnorm{M} +1\\
    \bnorm{M\isub{N}{x}} &=& \bnorm{M}.
  \end{eqnarray*}
\end{definition}

\begin{proposition}  \label{prop:box-term} The reduction $\rclsubredc$ terminates.
\end{proposition}
\begin{proof}
  The proof of the termination of the relation $\rclsubredc$ is based on the measure $\bnorm{\cdot}$ defined in Definition~\ref{def-measure}.
  We associate with each term $M$ a multiset $\multsbt(M)$ of natural numbers as
  follows:
  \begin{eqnarray*}
    \multsbt(x) &=&  \{\!\!\{~\}\!\!\}\\
    \multsbt(`l y.M) &=&  \multsbt(M)\\
    \multsbt(M\;P) &=& \multsbt(M)\cup \multsbt(P)\\
    \multsbt(\weak{x}{M}) &=&  \multsbt(M)\\
    \multsbt(\cont{x}{y}{z}{M}) &=&  \multsbt(M)\\
    \multsbt(M\isub{N}{x}) &=&  \{\!\!\{\bnorm{M}\}\!\!\} \cup \multsbt(M)\\
  \end{eqnarray*}
  Notice that if a term $P$ does not contain any substitution, then $\multsbt(P)
  =\{\!\!\{~\}\!\!\}$.  The multiset order is defined for instance
  in~\cite{baadnipk98} and is denoted by $\gg$. The rules in Figure~\ref{fig:sub-rcl}
  yield the following inequalities.
  \begin{footnotesize}
    \begin{eqnarray*}
      \{\!\!\{\bnorm{x}\}\!\!\} &\gg& \multsbt(N)\\
      \{\!\!\{\bnorm{M}+1\}\!\!\}  \cup \multsbt(M) &\gg& \{\!\!\{\bnorm{M}\}\!\!\}  \cup
      \multsbt(M)\\
      \{\!\!\{\bnorm{M}+\bnorm{P}+1\}\!\!\}  \cup \multsbt(M) \cup \multsbt(P)&\gg& \{\!\!\{\bnorm{M}\}\!\!\}  \cup \multsbt(M)  \cup \multsbt(P)\\
      \{\!\!\{\bnorm{M}+\bnorm{P}+1\}\!\!\}  \cup \multsbt(M) \cup \multsbt(P)&\gg& \{\!\!\{\bnorm{P}\}\!\!\} \cup \multsbt(M) \cup \multsbt(P) \\
      \{\!\!\{\bnorm{M}+1\}\!\!\}  \cup \multsbt(M) &\gg& \{\!\!\{\bnorm{M}\}\!\!\}  \cup \multsbt(M) \\
      \{\!\!\{\bnorm{M}+1\}\!\!\}  \cup \multsbt(M) &\gg& \multsbt(M)\\
      \{\!\!\{\bnorm{M}+1\}\!\!\}  \cup \multsbt(M) &\gg& \{\!\!\{\bnorm{M}\}\!\!\}  \cup \multsbt(M) \\
      \{\!\!\{\bnorm{M}+1\}\!\!\}  \cup \multsbt(M) &\gg& \{\!\!\{\bnorm{M}, \bnorm{M}\}\!\!\}  \cup
      \multsbt(M) \cup  \multsbt(M)
    \end{eqnarray*}
  \end{footnotesize}

  Two inequalities require discussion. The first comes from $x\isub{N}{x}
  \rclsubred N$ and is satisfied because $N$ is substitution free, therefore
  $\multsbt(N) =\{\!\!\{~\}\!\!\}$.  The second comes from
  $(\cont{x}{x_1}{x_2}{M})\isub{N}{x} \rclsubred
  \cont{Fv(N)}{Fv(N_1)}{Fv(N_2)}{M\isub{N_1}{x_1}\isub{N_2}{x_2}}$ and is
  satisfied because $\bnorm{\cont{x}{x_1}{x_2}{M}} = \bnorm{M}+1$ is larger than
  $\bnorm{M}$ and than any $\bnorm{P}$ for $P$ subterm of $M$.

  This shows that $\rclsubredc$ is well-founded, hence that $\rclsubredc$ terminates.
\end{proof}

\begin{proposition}\label{prop:rclsubred-conf}
  The reduction $\rclsubredc$ is confluent.
\end{proposition}
\begin{proof}
  There is no superposition between the left-hand sides of the rules of
  Figure~\ref{fig:sub-rcl}, therefore there is no critical pair.  Hence, the rewrite
  system is locally confluent. According to Proposition~\ref{prop:box-term}  it terminates, hence it is confluent by Newman's Lemma~\cite{baadnipk98}.
\end{proof}

\begin{definition}[$\rclsubred$ Normal forms]\label{def:nf}
  Starting from $M$ and reducing by $\rclsubredc$, the irreducible term we obtain is
  called the $\rclsubred$-\emph{normal form} of $M$ and denoted by $M\subnf$.
\end{definition}

Every $\rclsub$-term has a unique normal form, the existence is guaranteed by Proposition~\ref{prop:box-term}, whereas the uniqueness  is a consequence of confluence (Proposition~\ref{prop:rclsubred-conf}).


\begin{proposition}\label{prop:exh}
If $Q \in \Rcl$ then $Q\isub{N}{x}\subnf \in \Rcl$.
\end{proposition}
\begin{proof}
  Let us look at all the terms of the form $Q\isub{N}{x}$ and their evaluation by the rules in Figure~\ref{fig:sub-rcl}.
  \begin{itemize}
  \item $Q$ is a \emph{variable}. Due to rule $(sub)$ in Figure~\ref{fig:wfsub}, $x\in
    \Fvsub(Q)$, hence $Q$ must be $x$. Therefore, all the cases when $Q$ is a variable are
    exhausted.
  \item $Q$ is an \emph{abstraction}, then one rule is enough.
  \item $Q$ is an \emph{application} $MP$, then either $x\in\Fvsub(M)$ or $x\in\Fvsub(P)$, hence the two rules exhaust this case.
  \item $Q$ is an \emph{erasure} $\weak{y}{M}$, then either $y=x$ or $y\neq x$ and the two
    cases are considered.
  \item $Q$ is a \emph{duplication} $\cont{x}{x_1}{x_2}{M}$, then again either $y=x$ or
    $y\neq x$ and the two cases are considered.
  \end{itemize}
\end{proof}

\begin{proposition}\label{prop:nf}
  If $M`:\Rclsub$ then $M\subnf`:\Rcl$.
  \begin{proof}
  By induction on the number of substitutions in $M$, Proposition~\ref{prop:exh} being the base case.
  \end{proof}
\end{proposition}

The substitution of $n$ different variables in the same term is denoted by \[M\isub{N_1}{x_1}...\isub{N_n}{x_n}\subnf.\]
These substitutions are actually performed in ``parallel'' since we prove that they commute in the following proposition.

\begin{proposition}\label{prop:sub-com}
 If $M \in \Rcl$ and $x_i \in Fv(M)$ for $i \in \{1,...,n\},\; n \geq 1$ with $x_i\neq x_j$ for $i \neq j$, then
 $$M\isub{N_1}{x_1}...\isub{N_n}{x_n}\subnf \ =\ M\isub{N_{p(1)}}{x_{p(1)}}...\isub{N_{p(n)}}{x_{p(n)}}\subnf,$$
 where $(p(1),...,p(n))$ is a permutation of $(1,...,n)$.
\end{proposition}

\begin{proof}
We prove the proposition by induction on the structure of $M$,
  \begin{itemize}
  \item

For $M = x_1$ the statement holds since the only permutation is the identity, namely, $p(1)=1$, therefore  $x_1\isub{N_1}{x_1} \subnf = N_1 =x_1\isub{N_{p(1)}}{x_{p(1)}} \subnf $.

  \item If $M = `ly.Q$ then this works by induction. Notice that $y \neq x_i$, for $i \in \{1,...,n\}$.

  \item If $M = Q R$ then we distinguish two cases:
  \begin{itemize}
  \item
  some of $\{x_1,...,x_n\}$ belong to $Fv(Q)$, whereas the others belong to $Fv(R)$. Without  loss of generality we can assume that for some $k$ such that $1 \leq k < n$,  $\{x_1,...,x_k\}\in Fv(Q)$ and $\{x_{k+1},...,x_n\}\in Fv(R)$. Then
    $(Q R)\isub{N_1}{x_1}...\isub{N_n}{x_n}$ reduces to\\ $Q\isub{N_1}{x_1}...\isub{N_k}{x_k}\subnf R\isub{N_{k+1}}{x_{k+1}}...\isub{N_n}{x_n}\subnf$, and the result follows by two applications of induction hypothesis.
  \item If $M = Q R$ and $\{x_1,...,x_n\}$ all belong to either $Fv(Q)$ or to $Fv(R)$, the result follows by induction.
  \end{itemize}

  \item If $M = \weak{y}{Q}$ with $y\neq x_i$ for $i \in \{1,...,n\}$, then the result follows by
    induction.
  \item If $M = \weak{x_j}{Q}$ then $(\weak{x_j}{Q})\isub{N_1}{x_1}...\isub{N_j}{x_j}...\isub{N_n}{x_n}$ reduces to\\ $\weak{Fv(N_j)}{Q\isub{N_1}{x_1}...\isub{N_{j-1}}{x_{j-1}}\isub{N_{j+1}}{x_{j+1}}...\isub{N_n}{x_n}}\subnf$.\\
      On the other hand, given an arbitrary permutation $p$, let us call $k$ the index such that $p(k)=j$. Then,
   $(\weak{x_j}{Q})\isub{N_{p(1)}}{x_{p(1)}}...\isub{N_{p(k)}}{x_{p(k)}}...\isub{N_{p(n)}}{x_{p(n)}}$
      reduces to \\
$\weak{Fv(N_{p(k)})}{Q\isub{N_{p(1)}}{x_{p(1)}}...\isub{N_{p(k)-1}}{x_{p(k)-1}}\isub{N_{p(k)+1}}{x_{p(k)+1}}...\isub{N_{p(n)}}{x_{p(n)}}}\subnf$.
    Since $N_j=N_{p(k)}$ then $Fv(N_j)=Fv(N_{p(k)})$ and  the result follows by induction.
  \item If $M = \cont{y}{y_1}{y_2}{Q}$ where $y\neq x_i$ for $i \in \{1,...,n\}$, then the result follows by induction.
  \item
    If $M = \cont{x_j}{x'_j}{x''_j}{Q}$ then $(\cont{x_j}{x'_j}{x''_j}{Q})\isub{N_1}{x_1}...\isub{N_j}{x_j}...\isub{N_n}{x_n}$ reduces to\\
    $\cont{Fv(N_j)}{Fv(N'_j)}{Fv(N''_j)}{Q\isub{N_1}{x_1}...\isub{N'_j}{x'_j}\isub{N''_j}{x''_j}...\isub{N_n}{x_n}} \equiv M_1$.\\
    On the other hand, given an arbitrary permutation $p$, let us call $k$
      the index such that $p(k)=j$. We have that\\ $(\cont{x_{p(k)}}{x'_{p(k)}}{x''_{p(k)}}{Q})\isub{N_{p(1)}}{x_{p(1)}}...\isub{N_{p(k)}}{x_{p(k)}}...\isub{N_{p(n)}}{x_{p(n)}}$ reduces to\\
      $\cont{Fv(N_{k})}{Fv(N'_{k})}{Fv(N''_{k})}{Q\isub{N_{p(1)}}{x_{p(1)}}...\isub{N'_{p(k)}}{x'_{p(k)}}\isub{N''_{p(k)}}{x''_{p(k)}}...
        \isub{N_{p(n)}}{x_{p(n)}}} \equiv M_2$.
      By induction hypothesis (recall that $j=p(k)$),\\
      \hspace*{30pt}  $Q\isub{N_1}{x_1}...\isub{N'_{j}}{x'_{j}}\isub{N''_{j}}{x''_{j}}...\isub{N_n}{x_n}$
      and \\
      \hspace*{30pt} $Q\isub{N_{p(1)}}{x_{p(1)}}...\isub{N'_{p(k)}}{x'_{p(k)}}\isub{N''_{p(k)}}{x''_{p(k)}}...
      \isub{N_{p(n)}}{x_{p(n)}}$ \\
      have the same normal forms, therefore $M_1\subnf =
      M_2\subnf$.
  \end{itemize}
\end{proof}

Finally, we can formally define \emph{substitution} in $\Rcl$ and \emph{simultaneous
  substitution}  in $\Rcl$  via $\rclsub$-normal forms.

\begin{definition}[Substitution in $\Rcl$ ]
  If $M`:\Rcl$ and $N`:\Rcl$ then $$M\ISUB{N}{x} \triangleq M\isub{N}{x}\subnf.$$
\end{definition}

Notice that $M\ISUB{N}{x}$ is well-defined, since $M\ISUB{N}{x} \in \Rcl$, due to Proposition~\ref{prop:nf}.
Moreover, Proposition~\ref{prop:sub-com} allows us to give simply a meaning to simultaneous
  substitution.
\begin{definition}[Simultaneous substitution  in $\Rcl$ ]
Simultaneous  substitution  in $\Rcl$  is defined as follows:
\begin{displaymath}
  M\ISUBM{N_1}{x_1}{N_p}{x_p} = M\ISUB{N_1}{x_1}...\ISUB{N_p}{x_p}.
\end{displaymath}
provided that $Fv(N_i)\cap Fv(N_j)=\emptyset$ for $i \neq j$.
\end{definition}

\subsection{Operational semantics}

The operational semantics of $\rcl$ is defined by a \emph{reduction relation} $\rightarrow$, given by the set of reduction rules in Figure~\ref{fig:red-rcl}. In the $\rcl$-calculus, one works modulo the {\em structural equivalence} $\equiv_{\rcl}$, defined as the smallest equivalence that satisfies the axioms given in Figure~\ref{fig:equiv-rcl} and is closed under $\alpha$-conversion. The reduction relation $\rightarrow$ is closed under $\equiv_{\rcl}$ and contexts. Its reflexive, transitive closure will be denoted by $\rightarrowc$. As usual, a term is a {\em redex} if it has the form of a term on the left-hand side of a rule in Figure~\ref{fig:red-rcl}, whereas its {\em contractum} is the term on the right-hand side of the same rule.


\begin{figure}[htbp]
\centerline{ \framebox{ $
\begin{array}{rrcl}
(\beta)             & (\lambda x.M)N & \rightarrow & M\ISUB{N}{x} \\[1mm]
(\gamma_1)          & \cont{x}{x_1}{x_2}{(\lambda y.M)} & \rightarrow & \lambda y.\cont{x}{x_1}{x_2}{M} \\
(\gamma_2)          & \cont{x}{x_1}{x_2}{(MN)} & \rightarrow & (\cont{x}{x_1}{x_2}{M})N, \;\mbox{if} \; x_1,x_2 \not\in Fv(N)\\
(\gamma_3)          & \cont{x}{x_1}{x_2}{(MN)} & \rightarrow & M(\cont{x}{x_1}{x_2}{N}), \;\mbox{if} \; x_1,x_2 \not\in Fv(M)\\[1mm]
(\omega_1)          & \lambda x.(\weak{y}{M}) & \rightarrow & \weak{y}{(\lambda x.M)},\;x \neq y\\
(\omega_2)          & (\weak{x}{M})N & \rightarrow & \weak{x}{(MN)}\\
(\omega_3)          & M(\weak{x}{N}) & \rightarrow & \weak{x}{(MN)}\\[1mm]
(\gamma \omega_1)   & \cont{x}{x_1}{x_2}{(\weak{y}{M})} & \rightarrow & \weak{y}{(\cont{x}{x_1}{x_2}{M})},\;y \neq x_1,x_2\\
(\gamma \omega_2)   & \cont{x}{x_1}{x_2}{(\weak{x_1}{M})} & \rightarrow & M \ISUB{x}{x_2}
\end{array}
$ }} \caption{Reduction rules} \label{fig:red-rcl}
\end{figure}

The reduction rules are divided into four groups.
The main computational step is $\beta$-reduction.
The group of $(\gamma)$
reductions perform propagation of duplications into the expression.
Similarly, $(\omega)$ reductions extract erasures out of
expressions. This discipline allows us to optimise the computation
by delaying duplication of terms on the one hand, and by
performing erasure of terms as soon as possible on the other.
Finally, the rules in the $(\gamma\omega)$ group explain the
interaction between the explicit resource operators that are of
different nature. Notice that in the rule $(\gamma\omega_{2})$ the substitution  in $\Rcl$ is actually a syntactic variable replacement, i.e.,\ renaming.\footnote{We decided to use the same notation in order to introduce less different notations.}



\begin{figure}
\centerline{ \framebox{ $
\begin{array}{lrcl}
(`e_1) & \weak{x}{(\weak{y}{M})} & \equiv_{\rcl} &
\weak{y}{(\weak{x}{M})}\\
(`e_2) & \cont{x}{x_1}{x_2}{M} & \equiv_{\rcl} & \cont{x}{x_2}{x_1}{M}\\
(`e_3) & \cont{x}{y}{z}{(\cont{y}{u}{v}{M})} & \equiv_{\rcl} &
\cont{x}{y}{u}{(\cont{y}{z}{v}{M})}\\
(`e_4) & \cont{x}{x_1}{x_2}{(\cont{y}{y_1}{y_2}{M})} &
\equiv_{\rcl}  & \cont{y}{y_1}{y_2}{(\cont{x}{x_1}{x_2}{M})},\;\;
x \neq y_1,y_2, \; y \neq x_1,x_2
\end{array}
$ }} \caption{Structural equivalence} 
\label{fig:equiv-rcl}
\end{figure}


\begin{proposition}[Soundness of $\rightarrowc$]\label{prop:soundR}~
  \begin{itemize}
\item  For all terms $M$ and $N$ such that $M \rightarrow N$, if $M`:\Rcl$, then $N`:\Rcl$.
  \item For all terms $M$ and $N$ such that $M \rightarrowc N$, if $M`:\Rcl$, then $N`:\Rcl$.
  \end{itemize}
\end{proposition} %

In particular, in the case of $(\beta)$-reduction 
if $(`l x.M)N`:\Rcl$, then     \[M\ISUB{N}{x}=M\isub{N}{x}\subnf`:\Rcl\] by Proposition~\ref{prop:nf}.

No variable is lost during the computation, which is stated by the following
proposition.

\begin{proposition}[Preservation of free variables by $\rightarrowc$]~\\
\label{prop:pres-of-FV}
If $M \rightarrowc N$ then $Fv(M) = Fv(N).$
\end{proposition}
\begin{proof}
The proof is by case analysis on the reduction rules and uses
Proposition~\ref{prop:subsubject}~(ii).
\end{proof}


First, let us observe the structure of the
$\rcl$-normal forms, given by the following abstract syntax. As usually, a term is a normal for if it does not have any redex as subterm.

\begin{definition}[Set of Normal Forms]\label{def:setOfNF}
  The set $\NF$ of normals forms is generated by the following abstract syntax:  
\begin{eqnarray*}
 M_{nf} & ::= & \lambda
x.M_{nf}\,|\,\lambda x. \weak{x}{M_{nf}}\,|\, x M_{nf}^{1} \ldots
M_{nf}^{n}\, |\, \cont{x}{x_1}{x_2}{M_{nf}}\\
&&\qquad \mbox{in the last case~} M_{nf}\equiv P_{nf}Q_{nf},\;x_1 \in Fv(P_{nf}),\;x_2 \in Fv(Q_{nf})\\
E_{nf} & ::= & \weak{x}{M_{nf}}\,|\,\weak{x}{E_{nf}}
\end{eqnarray*}
\end{definition}
where $n \geq 0$. It is necessary to distinguish normal forms $E_{nf}$ separately
  because the term $\lambda x. \weak{y}{M_{nf}}$ is not a normal form,
  since $\lambda x. \weak{y}{M_{nf}} \to_{\omega_1} {\weak{y}\lambda
    x.M_{nf}}$. Also, in the last case the term  $\cont{x}{x_1}{x_2}{P_{nf}Q_{nf}},
  \mbox{where~}\;{x_1 \in Fv(P_{nf}),} {\;x_2 \in Fv(Q_{nf})}$ is not necessarily a normal form since $P_{nf}Q_{nf}$ can be a redex, in turn  $M_{nf} \equiv P_{nf}Q_{nf}$ guarantees that the application is a normal form.  

Next we define the set of strongly normalising terms $\SN$.

\begin{definition}[Strongly normalising terms]
The set of strongly normalising terms $\SN$ is defined as follows:
\begin{displaymath}
  \prooftree
  M`:\NF
  \justifies M`:\SN
  \endprooftree
  \qquad
  \prooftree
  `A N`:\Rcl\ . \ M \rightarrowc N \ "=>" \ N`:\SN\\
  \justifies  M`:\SN
  \endprooftree
\end{displaymath}
\end{definition}

\newcommand{\abs}{\textsf{Abs}}
\newcommand{\era}{\textsf{Era}}
\newcommand{\var}{\textsf{Var}}
\newcommand{\hd}{\textsf{App}}
\newcommand{\dup}{\textsf{Dup}}

\begin{lemma}
\label{lem:hft}
Every term has one of the following forms, where $n \geq 0 $:
$$
\begin{array}{rlrl}
(\abs) & \lambda x.N, \;  \; 
& (\abs\hd) & (\lambda x.N)PT_{1} \ldots T_{n}\\
(\var) & xT_{1} \ldots T_{n}  &  (\dup\hd) & (\cont{x}{x_{1}}{x_{2}}{N})T_{1} \ldots T_{n}\\
(\era) & \weak{x}{N}& (\era\hd) & (\weak{x}{N})PT_{1} \ldots T_{n} \\
 \end{array}\\
$$
\end{lemma}
\begin{proof}
These terms are well-formed according to Definition~\ref{def:rcl-terms} (we did not explicitly write the conditions, since we work with linear terms). The proof is by induction on the structure of the term $M \in \LR$.
\begin{itemize}
\item If $M$ is a variable, this case is covered by $\var$ for $n=0$.
\item If $M$ is an abstraction $\lambda x.Q$, then by induction $Q$ has one of the given forms, hence  $\lambda x.Q$ is covered by $\abs$.

\item If $M$ is an application then $M$ is of the form $M \equiv QP_{1} \ldots P_{n}$, for $n \geq 1$ and $Q$ is not an application. We proceed by subinduction on the structure of $Q$. Accordingly, $Q$ is one of the following:
	\begin{itemize}
	\item $Q$ is a variable, then we have the case $\var$, with $n \geq 1$;
	\item $Q$ is an abstraction, then we have the case $\abs\hd$;
	\item $Q$ is an erasure, then we have the case $\era\hd$;
	\item $Q$ is a duplication, then we have the case $\dup\hd$, with $n \geq 1$.
	\end{itemize}
\item If $M$ is an erasure $\weak{x}{Q}$, then by induction $Q$ has one of the given forms, hence $\weak{x}{Q}$  is covered by $\era$.
\item If $M$ is a duplication $\cont{x}{x_1}{x_2}{Q}$, then by induction $Q$ has one of the given forms, hence $\cont{x}{x_1}{x_2}{Q}$ is covered by $\dup\hd$ for $n=0$.
\end{itemize}
\end{proof}

\section{Intersection types for $\rcl$}
\label{sec:types}

In this section we introduce an intersection type assignment $\rcl \cap$ system which assigns \emph{strict types} to $\rcl$-terms. Strict
types were proposed in~\cite{bake92} and used
in~\cite{espiivetlika12} for characterisation of strong
normalisation in $\lG$-calculus.

The syntax of types is defined as follows:
$$
\begin{array}{lccl}
\textrm{Strict types}  & \tS & ::= & p \mid \tA \to \tS\\
\textrm{Types}      & \tA & ::= & \cap^n_i \tS_i
\end{array}
$$
\noindent where $p$ ranges over a denumerable set of type atoms and

$$\cap_i^n \tS_i = \left\{ \begin{array}{rr}
\tS_1 \cap \ldots \cap \tS_n & \mbox {for } n > 0\\
\top & \mbox{ for } n=0
\end{array}
\right.$$
\noindent $\top$ being the \emph{neutral element} for the intersection operator,
i.e.\ $`s\cap \top=`s$.

We denote types by $\tA,\tB,\tC...$, strict types by 
$\tS,\tT,\tU...$ and the set of all types by $\mathsf{Types}$. We
assume that the intersection operator is
commutative and associative.  We also assume that
intersection has priority over arrow. Hence, we will
omit parenthesis in expressions like $(\cap^n_i \tT_i) \to \tS$.

\subsection{The type assignment system}


\begin{definition}
  \begin{itemize}
  \item[(i)]  A \emph{basic type assignment (declaration)} is an expression of the form
  $x:\tA$, where $x$ is a term variable and $\tA$ is a type.
\item[(ii)] Consider a finite set $Dom(`G)$ of variables. A \emph{basis} is a
  function
  \begin{displaymath}
    `G: Dom(`G) "->" \mathsf{Types}.
  \end{displaymath}
  A \emph{basis extension} of $`G$ is a function $`G, x:`a : Dom(`G) \cup \{x\} "->"
  \mathsf{Types}$:
  \begin{displaymath}
    y "|->" \left\{
      \begin{array}{ll}
        `G(y) & \mbox{if~} y`:Dom(`G)\\
        `a & \mbox{if~} y=x
      \end{array}
    \right.
  \end{displaymath}
\item [(iii)] Given $`G$ and $`D$ such that $Dom(`G) = Dom(`D)$, the \emph{bases
    intersection} of $`G$ and $`D$ is the function $`G\sqcap `D: Dom(`G) "->"
  \mathsf{Types}$, such that:
  \begin{displaymath}
    `G\sqcap `D(x) = `G(x) \cap `D(x).
  \end{displaymath}
\item [(iv)] $\gtop$ is the constant function $\gtop: Dom(`G) "->" \{\top\}$.
  \end{itemize}
\end{definition}

In what follows we assume that the bases intersection has priority
over the basis extension, hence the parenthesis in
$\Gamma, (\Delta_1\sqcap \ldots \sqcap \Delta_n)$ will be omitted.
It is easy to show that $\gtop \sqcap \Delta = \Delta$ for
arbitrary bases $\Gamma$ and $\Delta$ that can be intersected,
hence $\gtop$ is the neutral element for the
intersection of bases of domain $Dom(\Gamma)$.

\begin{figure}[ht]
\centerline{ \framebox{ $
\begin{array}{c}
\\
\infer[(Ax)]{x:\tS \vdash x:\tS}{}\\ \\
\infer[(\to_I)]{\Gamma \vdash \lambda x.M:\tA \to \tS}
                        {\Gamma,x:\tA \vdash M:\tS} \quad\quad
\infer[(\to_E)]{\Gamma, \dztop \sqcap \Delta_1 \sqcap ...
\sqcap \Delta_n \vdash MN:\tS}
                    {\Gamma \vdash M:\cap^n_i \tT_i \to \tS & \Delta_0 \vdash N:\tT_0\; \ldots\; \Delta_n \vdash N:\tT_n}\\ \\
\infer[(Cont)]{\Gamma, z:\tA \cap \tB \vdash
\cont{z}{x}{y}{M}:\tS}
                    {\Gamma, x:\tA, y:\tB \vdash M:\tS} \quad\quad
\infer[(Thin)]{\Gamma, x:\top \vdash \weak{x}{M}:\tS}
                    {\Gamma \vdash M:\tS}\\
\end{array}
$ }} \caption{$\rcl \cap$: $\rcl$-calculus with intersection
types} \label{fig:typ-rcl-int}
\end{figure}

The type assignment system $\rcl \cap$ is given in Figure~\ref{fig:typ-rcl-int}.
It is syntax directed and the rules are context-splitting. The axiom $(Ax)$ ensures that void $\lambda$-abstraction  cannot be typed, i.e.\ in a typeable term each free variable appears at least once. The context-splitting rule $(\to_E)$ ensures that in a typeable term each free variable appears not more than once.

Assume that we implement these properties in the 
type system with $(Ax)$, $(\to_E)$ and
$(\to_I)$, 
then the combinators  $K=\lambda xy.x$ and $W^{-1} = \lambda
xy.xyy$ would not be typeable.  This motivates and justifies the introduction of the operators of erasure and duplication  and the corresponding typing rules $(Thin)$ and $(Cont)$, which further maintain the explicit control of resources and enable the typing of  $K$ and $W^{-1}$,  namely of their  corresponding $\rcl$-terms $\lambda xy.\weak{y}{x}$ and $ \lambda xy.\cont{y}{y_1}{y_2}xy_1y_2$, respectively. Let us mention that on the logical side, structural rules of {\em thinning} and {\em contraction}  are present in Gentzen's original formulation of $LJ$, Intuitionistic Sequent Calculus, but not in $NJ$, Intuitionistic Natural Deduction~\cite{gent34,gent35}. Here instead, the presence of the typing rules $(Thin)$ and $(Cont)$ completely maintains the explicit control of resources in $\rcl$.

In the proposed system, intersection
types occur only in two inference rules. In the rule $(Cont)$ the
intersection type is created, this being \emph{the only} place where
this happens. This is justified because it corresponds to the
duplication of a variable. In other words, the control of the
duplication of variables entails the control of the introduction of
intersections in building the type of the term in question. In the
rule $(\to_E)$, intersection appears on the right hand side of the turnstyle $"|-"$
which corresponds to the usage of the intersection type after it
has been created by the rule $(Cont)$ or by the rule $(Thin)$
if $n=0$.

The role of $`D_0$ in the rule $(\to_E)$ should be emphasized.  It is needed only when $n=0$
to ensure that $N$ has a type, i.e.\ that $N$
is strongly normalising as would be seen below. Then, in the 
conclusion of the rule, the types
of the free variables of $N$ can be forgotten, hence all the free
variables of $N$ receive the type $\top$.  All the free
variables of the term must occur in the environment $\Gamma$ (see Lemma~\ref{lem:dom-corr}),
therefore useless
variables occur with the type $\top$.  When $n>0$,  $`D_0$
can be any of the other environments and the type of $N$ the
associated type.  Since $`D^{\top}$ is a neutral element for $\sqcap$,
when $n>0$, $`D^{\top}$~ disappears in the 
conclusion of the rule.  The case $n=0$
resembles the rules $(drop)$ and/or \textit{(K-cup)}
in~\cite{lenglescdougdezabake04} and was used to present the two
cases, $n=0$ and $n\neq 0$ in a uniform way.  In the rule $(Thin)$ the
choice of the type of $x$ is~$\top$, since this corresponds to a
variable which does not occur anywhere in~$M$. The remaining rules,
namely $(Ax)$ and $(\to_I)$ are traditional, i.e.\ they are the same as in
the simply typed $\lambda$-calculus. Notice however that the type of
the variable in $(Ax)$ is a strict type.

\subsubsection*{Roles of the variables}
In the syntax of $\rcl$, there are three
kinds of variables according to the way they are introduced, namely as
a placeholder (associated with the typing rule \emph{(Ax)}), as the result of a duplication (associated with the typing rule
\emph{(Cont)}) or as the result of an erasure (associated with the typing rule \emph{(Thin)}). Each kind of variable receives a specific
type:
\begin{itemize}
\item variables as placeholders have a strict type,
\item variables resulting from a duplication have an intersection type,
\item variables resulting from an erasure have the type $\top$.
\end{itemize}

In order to emphasize the sensitivity of the system $\rcl \cap$ w.r.t.\
the role of a variable in a term, we provide the following examples in which
variables change their role during the computation process. Our goal is to show
that when the role of a variable changes, its type in the type derivation changes as well, so that the correspondence between particular roles and types is preserved.

\begin{example}
A variable as a ``placeholder'' becomes an ``erased'' variable:
this is the case with the variable $z$ in $(\lambda x.\weak{x}{y})z$, because
\begin{center}
$(\lambda x.\weak{x}{y})z \, \to_{\beta} \, (\weak{x}{y})\ISUB{z}{x}\, \triangleq \, (\weak{x}{y})\isub{z}{x}\subnf\, =\, \weak{z}{y}.$
\end{center}
Since $z:\top,y:\tS \vdash \weak{z}{y}:\tS$, we want to show that $z:\top,y:\tS \vdash (\lambda x.\weak{x}{y})z:\tS$.\\
Indeed:
\[\prooftree
    \prooftree
      \prooftree
           \prooftree
                \justifies
                y:\tS \vdash y:\tS
                \using{(Ax)}
           \endprooftree
           \justifies
           x:\top, y:\tS \vdash \weak{x}{y}:\tS
           \using{(Weak)}
        \endprooftree
    \justifies
    y:\tS \vdash \lambda x. \weak{x}{y}:\top \to \tS
    \using{(\to_I)}
    \endprooftree
    \prooftree
        \justifies
        z:\tT \vdash z:\tT
        \using{(Ax)}
    \endprooftree
    \justifies
     z:\top, y:\tS  \vdash (\lambda x. \weak{x}{y})z: \tS
    \using{(\to_E).}
\endprooftree
\]
In the rule $(\to_E)$, we have $n=0$,  $\Delta_0 = z:\tT$ and $\dztop =
z:\top$. Thus, in the previous derivation, the variable $z$ changed its type from a
strict type to $\top$, in accordance with the change of its role in the bigger term.
\end{example}

\begin{example}
A variable as a ``placeholder'' becomes a ``duplicated'' variable:
this is the case with the variable $v$ in $(\lambda x.\cont{x}{y}{z}{yz})v$, because
\begin{center}
$(\lambda x.\cont{x}{y}{z}{yz})v \,\to_{\beta}\,(\cont{x}{y}{z}{yz})\ISUB{v}{x} \,\triangleq (\cont{x}{y}{z}{yz})\isub{v}{x}\subnf \,=$\\
$= \,\cont{Fv[v]}{Fv[v_1]}{Fv[v_2]}{(yz)\isub{v_1}{y}\isub{v_2}{z}}\subnf\, =\, \cont{v}{v_1}{v_2}{v_1v_2}.$
\end{center}
Since $v:(\tT \to \tS)\cap \tT \vdash \cont{v}{v_1}{v_2}{v_1v_2}:\tS$, we want to show that \\
$v:(\tT \to \tS)\cap \tT  \vdash (\lambda x.\cont{x}{y}{z}{yz})v:\tS$.\\
Indeed:
\[\prooftree
    \prooftree
    \vdots
    \justifies
    \vdash \lambda x.\cont{x}{y}{z}{yz}:((\tT \to \tS)\cap \tT) \to \tS
    \using{(\to_I)}
    \endprooftree
    \prooftree
        \justifies
        v:\tT \vdash v:\tT
        \using{(Ax)}
    \endprooftree
    \prooftree
        \justifies
        v:\tT \to \tS \vdash v:\tT \to \tS
        \using{(Ax)}
    \endprooftree
    \justifies
     v:(\tT \to \tS) \cap \tT   \vdash (\lambda x.\cont{x}{y}{z}{yz})v: \tS
    \using{(\to_E).}
\endprooftree
\]
In the rule $(\to_E)$, we have $n=2$, therefore $\Delta_0 \vdash N: \tT_0$ can be one of the two existing typing judgements, for instance $v:\tT \vdash v:\tT$. In this case $\dztop$ disappears in the conclusion, because \\
$\dztop \sqcap \Delta_1 \sqcap \Delta_2 = v:\top \sqcap v:\tT \to \tS \sqcap v:\tT = v:\top\cap(\tT \to \tS)\cap\tT = v:(\tT \to \tS) \cap \tT$. Again, we see that the type of the variable $v$ changed from strict type to (intersection) type.
\end{example}

\begin{example}
A ``duplicated'' variable becomes an ``erased'' variable:
this is the case with the variable $z$ in $(\lambda x.\weak{x}{y})(\cont{z}{u}{v}{uv})$, because
\begin{center}
$(\lambda x.\weak{x}{y})(\cont{z}{u}{v}{uv})\, \to_{\beta}\,(\weak{x}{y})\ISUB{\cont{z}{u}{v}{uv}}{x}\, \,\triangleq \,(\weak{x}{y})\isub{\cont{z}{u}{v}{uv}}{x}\subnf \,= $\\
$= \,\weak{Fv(\cont{z}{u}{v}{uv})}{y}\, =\, \weak{z}{y}.$
\end{center}
Like in the previous examples,  both $z:\top,y:\tS \vdash \weak{z}{y}:\tS$ and $z:\top,y:\tS \vdash (\lambda x.\weak{x}{y})(\cont{z}{u}{v}{uv}):\tS$ can be shown.
\end{example}

\begin{example}
An ``erased'' variable becomes a ``duplicated'' variable:
this is the case with the variable $u$ in $(\lambda x.\cont{x}{y}{z}{yz})(\weak{u}{v})$, because
\begin{eqnarray*}
(\lambda x.\cont{x}{y}{z}{yz})(\weak{u}{v}) &\to_{\beta}& (\cont{x}{y}{z}{yz})\ISUB{\weak{u}{v}}{x}\\
& \triangleq & (\cont{x}{y}{z}{yz})\isub{\weak{u}{v}}{x}\subnf\\
&=& \cont{Fv[\weak{u}{v}]}{Fv[\weak{u_1}{v_1}]}{Fv[\weak{u_2}{v_2}]}{(yz)\isub{\weak{u_1}{v_1}}{y}\isub{\weak{u_2}{v_2}}{z}}\subnf\\
 &=& \cont{u}{u_1}{u_2}{\cont{v}{v_1}{v_2}{(\weak{u_1}{v_1})(\weak{u_2}{v_2})}}.
\end{eqnarray*}
The situation here is slightly different. Fresh variables $u_1$ and $u_2$ are
obtained from $u$ using the substitution  in $\Rcl$ . The variable $u$ is
introduced by thinning, so its type is $\top$. Substitution  in $\Rcl$ does
not change the types, therefore both $u_1$ and $u_2$ have the type ~$\top$. Finally,
$u$ in the resulting term is obtained by contracting $u_1$ and $u_2$, therefore its
type is $\top \cap \top = \top$. Thus we have an interesting situation - the role of the variable $u$ changes from ``to be erased'' to ``to be duplicated'', but its type remains $\top$.

However, this paradox (if any) is only apparent, as well as the change of
the role. Unlike the previous three examples, in which we obtained normal forms, in this case the computation can continue:
\begin{eqnarray*}
\cont{u}{u_1}{u_2}{\cont{v}{v_1}{v_2}{(\weak{u_1}{v_1})(\weak{u_2}{v_2})}} &\to_{(\omega_2 + `e_4)}& \cont{v}{v_1}{v_2}{\cont{u}{u_1}{u_2}{\weak{u_1}{v_1(\weak{u_2}{v_2})}}}\\
&\to_{\gamma\omega_2}& \cont{v}{v_1}{v_2}{v_1((\weak{u_2}{v_2}))\ISUB{u}{u_2}}\\
& =& \cont{v}{v_1}{v_2}{v_1(\weak{u}{v_2})}.
\end{eqnarray*}
So, we see that the actual role of the variable $u$  in the obtained normal form, is ``to be erased'', as indicated by its type $\top$.

\end{example}

To conclude the analysis, we point out the following key points:
\begin{itemize}
\item The type assignment system $\rcl\cap$ is constructed in such way that the type
  of a variable always indicates its actual role in the term. Due to this, we claim that the system $\rcl\cap$ fits naturally to the
  resource control calculus $\rcl$.
\item Switching between roles is not reversible: once a variable is meant to be
  erased, it cannot be turned back to some other role. Moreover, the information about its former role cannot be reconstructed from the type.
\end{itemize}

\subsubsection*{A note about idempotence and identity rule} Recall that the typing tree of a term is dictated by the syntax: $\to$ is
  introduced by $(\to_I)$, $\cap$ is introduced by $(Cont)$ and $\top$ is
  introduced by $(Thin)$. In this context it would not pertain to remove an
  intersection by idempotence or identity rule.  This is why they are not considered here.

\subsection{Structural properties}

\begin{lemma}[Domain correspondence for $\rcl\cap$]
\label{lem:dom-corr}
Let $`G \vdash M:\tS$ be a typing judgment. Then ${x`: Dom(`G)}$
if and only if $x`: Fv(M)$.
\end{lemma}
\begin{proof}
  The rules of Figure~\ref{fig:typ-rcl-int} belong to three categories.
  \begin{enumerate}
  \item \emph{The rules that introduce a variable}. These rules are
    \emph{(Ax)}, $(Cont)$ and $(Thin)$.  One sees that the variable is
    introduced in the environment if and only it is introduced in the
    term as a free variable.
  \item \emph{The rules that remove variables}.  These rules are
    $(\to_I)$ and $(Cont)$.   One sees that the variables are
    removed from the environment if and only if they are removed  from the
    term as a free variable.
  \item
  \emph{The rule that neither introduces nor removes a variable.} This rule is $(\to_E)$.
  \end{enumerate}
Notice that $(Cont)$ introduces and removes variables.
\end{proof}
The Generation Lemma makes somewhat more precise the Domain
Correspondence Lemma.

\begin{lemma}[Generation lemma for $\rcl\cap$]
\label{prop:intGL} \rule{0in}{0in}
\begin{sloppypar}
  \begin{enumerate}
  \item[(i)]\quad $\Gamma \vdash \lambda x.M:\tT\;\;$ iff there exist $`a$ and $`s$
    such that $\;\tT\equiv \tA\rightarrow \tS\;\;$ and ${\Gamma,x:\tA \vdash M:\tS.}$
  \item[(ii)]\quad $\Gamma \vdash MN:\tS\;\;$ iff and there exist $`D_i$ and
    $\tT_i,\;i \in \{0, \ldots, n\}$ such that ${\Gamma' \vdash M:\cap_{i}^{n}\tT_i\to \tS}$
    and for all $i \in \{0, \ldots, n\}$, ${\Delta_{i} \vdash N:\tT_i}$ and
    ${\Gamma=\Gamma', \dztop \sqcap \Delta_{1} \sqcap \ldots \sqcap \Delta_{n}}$.
  \item[(iii)]\quad $\Gamma \vdash \cont{z}{x}{y}{M}:\tS\;\;$ iff there exist $`G',
    `a, `b$ such that
    $\;\Gamma=\Gamma',z:`a\cap`b$
    and $\;\Gamma', x:\tA, y:\tB \vdash M:\tS.$
  \item[(iv)]\quad $\Gamma \vdash \weak{x}{M}:\tS\;\;$ iff $\;\Gamma=\Gamma', x:\top$
    and $\;\Gamma' \vdash M:\tS.$
  \end{enumerate}
\end{sloppypar}

\end{lemma}

\begin{proof} The proof is straightforward since all the rules are syntax directed, and relies on Lemma~\ref{lem:dom-corr}.
\end{proof}

In the sequel, we prove that the proposed system satisfies the following properties: Substitution lemma for $\rcl\cap$ (Proposition~\ref{prop:sub-lemma}) and Subject reduction and equivalence (Proposition~\ref{prop:sr}).




In order to prove the Substitution lemma we
extend the type assignment system $\rcl\cap$ with a new rule for typing the substitution operator, thus obtaining an auxiliary system $\rclsub\cap$ that assigns types to $\rclsub$-terms.
\begin{definition}
\label{def:rcl-sub-typ}
\begin{itemize}
\item[(i)] The type assignment system $\rclsub\cap$ consists of rules from Figure~\ref{fig:typ-rcl-int} plus the following $(Subst)$ rule:
\begin{displaymath}
\infer[(Subst)]{\Gamma, \dztop \sqcap \Delta_1 \sqcap ... \sqcap
\Delta_n \vdashsub M\isub{N}{x}:\tS}
                    {\Gamma, x:\cap^n_i \tT_i \vdashsub M:\tS & \Delta_0 \vdash N:\tT_0 & ... & \Delta_n \vdash N:\tT_n}
\end{displaymath}
\item[(ii)] Typing judgements in the system $\rclsub\cap$ are denoted by $\Gamma \vdashsub M:\tS$.
\end{itemize}
\end{definition}

The system $\rclsub\cap$ is also syntax-directed, and assigns strict types to $\rclsub$-terms. Therefore, it represents \emph{a conservative extension} of the system $\rcl\cap$, meaning that if $\Gamma \vdashsub M:\tS$ and $M \in \Rcl$ (i.e. $M$ is substitution-free), then $\Gamma \vdash M:\tS$ and the two derivations coincide.

It is easy to adapt Lemma~\ref{lem:dom-corr} and Lemma~\ref{prop:intGL} to prove the corresponding properties of the system $\rclsub\cap$.
\begin{lemma}[Domain correspondence for $\rclsub\cap$]
\label{lem:dom-corr-sub}
Let $`G \vdashsub M:\tS$ be a typing judgment. Then ${x`: Dom(`G)}$
if and only if $x`: \Fvsub(M)$.
\end{lemma}
\begin{proof}
The proof is the same as the proof of Lemma~\ref{lem:dom-corr}, having in mind the definition of $\Fvsub(M)$ and the fact that the rule $(Subst)$ belongs to the category of rules that remove variables.
\end{proof}
\begin{lemma}[Generation lemma for $\rclsub\cap$]
\label{prop:intGL-sub}
\begin{enumerate}
\item[(i)]\quad $\Gamma \vdashsub \lambda x.M:\tT\;\;$ iff there
exist $`a$ and $`s$ such that $\;\tT\equiv \tA\rightarrow \tS\;\;$
and $\;\Gamma,x:\tA \vdashsub M:\tS.$
\item[(ii)]\quad $\Gamma \vdashsub MN:\tS\;\;$ iff
there exist $`D_i$ and $\tT_i,\;i  = 0, \ldots, n$ such that
$\Gamma' \vdashsub M:\cap_{i}^{n}\tT_i\to \tS$ and for all $i \in
\{0, \ldots, n\}$, $\;\Delta_{i} \vdashsub N:\tT_i$ and $\;\Gamma=\Gamma',
\dztop \sqcap \Delta_{1} \sqcap \ldots \sqcap \Delta_{n}$.
\item[(iii)]\quad $\Gamma \vdashsub \cont{z}{x}{y}{M}:\tS\;\;$ iff
there exist $`G', `a, `b$ such that
$\;\Gamma=\Gamma',z:`a\cap`b$ \\
and $\;\Gamma', x:\tA, y:\tB \vdashsub M:\tS.$
\item[(iv)]\quad $\Gamma \vdashsub \weak{x}{M}:\tS\;\;$ iff
$\;\Gamma=\Gamma', x:\top$ and $\;\Gamma' \vdashsub M:\tS.$
\item[(v)]\quad $\Gamma \vdashsub M\isub{N}{x}:\tS\;\;$ iff
there exist $`D_i$ and $\tT_i,\;i  = 0, \ldots, n$ such that
$\Gamma', x:\cap_{i}^{n}\tT_i \vdashsub M:\tS$ and for all $i \in
\{0, \ldots, n\}$, $\;\Delta_{i} \vdash N:\tT_i$ and $\;\Gamma=\Gamma',
\dztop \sqcap \Delta_{1} \sqcap \ldots \sqcap \Delta_{n}$.
\end{enumerate}
\end{lemma}
\begin{proof} The proof is straightforward since all the rules are syntax directed, and relies on Lemma~\ref{lem:dom-corr-sub}.
\end{proof}


To prove Lemma~\ref{lem:sub-typpres} we will need the definition of contexts.

\begin{definition}[$\rclsub$-Contexts]
\label{def:cntx}
$$\cntx ::= [\;] \; |\; \lambda x.\cntx \; |\; M\cntx \; |\; \cntx M \; |\; \weak{x}{\cntx} \; |\; \cont{x}{x_1}{x_2}{\cntx} \; |\; \cntx \isub{N}{x}$$
\end{definition}

\begin{lemma}[Type preservation under $\rclsubredc$]
\label{lem:sub-typpres}
\rule{0in}{0in}

\begin{itemize}
\item[(i)]
For all $M,M' \in \Rclsub$, $N \in \Rcl$, if $\Gamma \vdashsub M\isub{N}{x}:\tS$ and $M\isub{N}{x} \rclsubredc M'$, then $\Gamma \vdashsub M':\tS$.
\item[(ii)]
For all $M,M' \in \Rclsub$, $N \in \Rcl$, if $\Gamma \vdashsub \cntx[M\isub{N}{x}]:\tS$ and $\cntx[M\isub{N}{x}] \rclsubredc \cntx[M']$, then $\Gamma \vdashsub \cntx[M']:\tS$.
\end{itemize}

\end{lemma}

\begin{proof}

(i)
\noindent
The proof is by case analysis on $\rclsubred$ (Figure~\ref{fig:sub-rcl}).
We consider only some representative rules. The other rules are routine and their
proofs are analogous to the second rule we consider.
\begin{itemize}
\item Rule $x\isub{N}{x}  \rclsubred  N$.  In this case $n=1$ and $`G$ is
  empty. Recall that ${`D^{\top}\sqcap \Delta = \Delta}$.
On one hand we have
  \begin{displaymath}
    \infer[(Subst)]{\Delta \vdashsub x\isub{N}{x}: \tT}{\infer[(Ax)]{x: \tT \vdashsub x: \tT}{} & \Delta \vdashsub N:\tT & \Delta \vdashsub N:\tT}
  \end{displaymath}
and on the other hand we have
\begin{displaymath}
  `D\vdashsub N:`t
\end{displaymath}
by assumption.

\item Rule $(MP)\isub{N}{x}  \rclsubred   M\isub{N}{x} P, \;\;x \in \Fvsub(M)$. On
  one hand we have:
  \begin{footnotesize}
    \begin{displaymath}
     \hspace*{-25pt}
\parbox{\textwidth}{ 
      \infer[(Subst)]{\Gamma, `Q^\top_0 \sqcap `Q_1 \sqcap ... \sqcap
        `Q_m ,`D^\top_0 \sqcap `D_1 \sqcap ... \sqcap
        `D_n \vdashsub (M\;P)\isub{N}{x}:\tS}{\infer[\to_E]{\Gamma, x:\cap_{i}^{n}\tU_{i}, `Q^\top_0\sqcap `Q_1\sqcap
          ... \sqcap`Q_m \vdashsub M\;P:\tS}{\Gamma, x:\cap_{i}^{n}\tU_{i} \vdashsub M:\cap^m_i `r_i \to \tS & `Q_0
          \vdashsub P:`r_0 \ldots\; `Q_m \vdashsub P:`r_m} & \Delta_0 \vdashsub
        N:\tT_0 & ... & \Delta_n \vdashsub N:\tT_n}
}
    \end{displaymath}
  \end{footnotesize}

One the other hand we have:
  \begin{footnotesize}
    \begin{displaymath}
     \hspace*{-25pt}
\parbox{\textwidth}{ 
\infer[(\to_E)]{\Gamma, `Q^\top_0 \sqcap `Q_1 \sqcap ... \sqcap
        `Q_m ,`D^\top_0 \sqcap `D_1 \sqcap ... \sqcap
        `D_,n \vdashsub M\isub{N}{x}\;P:\tS}{\infer[(Subst)]{\Gamma, `D^\top_0\sqcap `D_1\sqcap
          ... \sqcap`D_n \vdashsub M\isub{N}{x}:\cap^m_i `r_i \to \tS}{\Gamma, x:\cap_{i}^{n}\tU_{i} \vdashsub M:\cap^m_i `r_i \to \tS& \Delta_0 \vdashsub N:\tT_0 & ... & \Delta_n \vdashsub N:\tT_n}  & `Q_0
        \vdashsub P:`r_0 \ldots\; `Q_m \vdashsub P:`r_m}
}
    \end{displaymath}
  \end{footnotesize}

\item Rule $(\weak{x}{M})\isub{N}{x}  \rclsubred   \weak{Fv(N)}{M}$.  In this case
  $n=0$. On one hand we have:
  \begin{footnotesize}
    \begin{displaymath}
      \prooftree
      \prooftree
      `G \vdashsub M:`s
      \justifies `G, x: \top \vdashsub \weak{x}{M}
      \using (Thin)
      \endprooftree
      `D_0 \vdashsub N:`t_0
      \justifies `G, `D^\top_0 \vdashsub (\weak{x}{M})\isub{N}{x}:`s
      \using (Subst)
      \endprooftree
    \end{displaymath}
  \end{footnotesize}

On the other hand we have:
  \begin{displaymath}
    \prooftree
    \prooftree
    `G \vdashsub M:`s
    \justifies \vdots
    \using (Thin)
    \endprooftree
    \justifies `G, `D^\top_0 \vdashsub \weak{Fv(N)}{M}:`s
    \using (Thin)
    \endprooftree
  \end{displaymath}

\item Rule $(\cont{x}{x_1}{x_2}{M})\isub{N}{x} \rclsubred
  \cont{Fv[N]}{Fv[N_1]}{Fv[N_2]}{M\isub{N_1}{x_1}\isub{N_2}{x_2}}$.
In order to make the proof tree readable, we adopt the following abbreviations:  
\begin{eqnarray*}
`t_1 &\triangleq& \cap_{i}^{n_1}`t_{1,i}\\
`t_2 &\triangleq& \cap_{i}^{n_2}`t_{2,i}\\
`D_1 &\triangleq& `D_{1,1}\sqcap ... \sqcap`D_{1,n_1}\\
`D_2 &\triangleq& `D_{2,1}\sqcap ... \sqcap`D_{2,n_2}\\
\mathfrak{L}_1 &\triangleq&`D_{1,1} \vdashsub N:`t_{1,1} \ ...\ `D_{1,n_1} \vdashsub N:`t_{1,n_1}\\
\mathfrak{L}_2 &\triangleq& `D_{2,1} \vdashsub N:`t_{2,1} ... `D_{2,n_2} \vdashsub N:`t_{2,n_2}
\end{eqnarray*}
Since $N_1$ and $N_2$ are obtained from $N$ only by renaming the free variables with fresh variables of the same type, for each derivation $`D_{1,i} \vdashsub N:`t_{1,i}$ where $i \in \{1,...,n_1\}$ we have $`D'_{1,i} \vdashsub N_1:`t_{1,i}$, where $`D_{1,i}$ and $`D'_{1,i}$ differ only in variables names. Analogously, for each derivation $`D_{1,j} \vdashsub N:`t_{1,j}$ where $i \in \{1,...,n_2\}$ we have $`D''_{1,j} \vdashsub N_2:`t_{1,j}$, where $`D_{1,j}$ and $`D''_{1,j}$ differ only in variables names. Now, we also adopt the following abbreviations: 
\begin{eqnarray*}
`D'_1 &\triangleq& `D'_{1,1}\sqcap ... \sqcap`D'_{1,n_1}\\
`D''_2 &\triangleq& `D''_{2,1}\sqcap ... \sqcap`D''_{2,n_2}\\
\mathfrak{L'}_1 &\triangleq&`D'_{1,1} \vdashsub N_1:`t_{1,1} \ ...\ `D'_{1,n_1} \vdashsub N_1:`t_{1,n_1}\\
\mathfrak{L''}_2 &\triangleq& `D''_{2,1} \vdashsub N_2:`t_{2,1} ... `D''_{2,n_2} \vdashsub N_2:`t_{2,n_2}
\end{eqnarray*}

Moreover, we do not consider the environment $`D_0$ since it is useless here.
Now, on one hand we have:
\begin{displaymath}
  \prooftree
  \prooftree
  `G, x_1:`t_1, x_2:`t_2 \vdashsub M:`s
  \justifies `G, x:`t_1\cap`t_2 \vdashsub \cont{x}{x_1}{x_2}{M}: `s
  \using (Cont)
  \endprooftree
   \mathfrak{L}_1 \quad \mathfrak{L}_2
  \justifies `G, `D_1\sqcap`D_2 \vdashsub (\cont{x}{x_1}{x_2}{M})\isub{N}{x}: `s
  \using (Subst)
  \endprooftree
\end{displaymath}
On the other hand we have
\begin{displaymath}
  \prooftree
  \prooftree
  \prooftree
  \prooftree
  `G, x_1:`t_1, x_2:`t_2 \vdashsub M:`s
  \quad \mathfrak{L'}_1
  \justifies `G ,`D'_1, x_2:`t_2 \vdashsub M\isub{N_1}{x_1}:`s
  \using (Subst)
  \endprooftree
  \quad \mathfrak{L''}_2
  \justifies `G, `D'_1, `D''_2 \vdashsub M\isub{N_1}{x_1}\isub{N_2}{x_2}:`s
  \using (Subst)
  \endprooftree
  \justifies \vdots
  \using (Cont)
  \endprooftree
  \justifies `G, `D_1\sqcap`D_2 \vdashsub \cont{Fv[N]}{Fv[N_1]}{Fv[N_2]}{M\isub{N_1}{x_1}\isub{N_2}{x_2}}:`s
  \using (Cont)
  \endprooftree
\end{displaymath}

\end{itemize}

(ii)
We will denote by $Q \equiv \cntx[M\isub{N}{x}]$ and $Q' \equiv \cntx[M']$. If $Q \rclsubredc Q'$ this means that $M\isub{N}{x} \rclsubredc M'$. We prove the statement by induction on the structure of a context containing a redex. We provide the proof for the basic case $\cntx = [\; ]$ and three additional cases $\cntx = \lambda x.\cntx'$,
$\cntx = \weak{x}{\cntx'}$ 
and $\cntx = \cntx' \isub{P}{y}$,
the proof being similar for the remaining context kinds.
\begin{itemize}

\item
Case $\cntx = [\;]$. This is the first part of this lemma (i).

\item
Case $\cntx = \lambda x.\cntx'$. Then $Q = \lambda x.\cntx'[M\isub{N}{x}]$ and $Q' = \lambda x.\cntx'[M']$.
By assumption $\Gamma \vdashsub Q:\tS$, i.e.\ $\Gamma \vdashsub \lambda x.\cntx'[M\isub{N}{x}]:\tS$. Using Generation lemma for $\rclsub\cap$ (Lemma~\ref{prop:intGL-sub}(i)) we obtain that there exist $\tA$ and $\tT$ such that $\tS = \tA \to \tT$ and $\Gamma, x:\tA \vdashsub \cntx'[M\isub{N}{x}]:\tT$. Since $M\isub{N}{x} \rclsubredc M'$ by IH we have that $\Gamma, x:\tA \vdashsub \cntx'[M']:\tT$. Using rule $(\to_{I})$ we can conclude that $\Gamma \vdashsub \lambda.\cntx'[M']:\tA \to \tT = \tS$.

\item
Case $\cntx = \weak{x}{\cntx'}$. Then $Q = \weak{x}{\cntx'[M\isub{N}{x}]}$ and $Q' = \weak{x}{\cntx'[M']}$.
By assumption $\Gamma \vdashsub Q:\tS$, i.e.\ $\Gamma \vdashsub \weak{x}{\cntx'[M\isub{N}{x}]}:\tS$. Using Generation lemma for $\rclsub\cap$ (Lemma~\ref{prop:intGL-sub}(iv)) we obtain that $\Gamma = \Gamma', x:\top$ and $\Gamma' \vdashsub \cntx'[M\isub{N}{x}]:\tS$. Since $M\isub{N}{x} \rclsubredc M'$  
by IH we have that $\Gamma' \vdashsub \cntx'[M']:\tS$. Using rule $(Thin)$ we can conclude that $\Gamma \vdashsub \weak{x}{\cntx'[M']}:\tS$.

\item 
Case $\cntx =\cntx' \isub{P}{y}$. Then $Q = \cntx'\isub{P}{y}[M\isub{N}{x}]$ and $Q' = \cntx'\isub{P}{y}[M']$.
By assumption $\Gamma \vdashsub Q:\tS$, i.e.\ $\Gamma \vdashsub \cntx'\isub{P}{y}[M\isub{N}{x}]:\tS$. Using Generation lemma for $\rclsub\cap$ (Lemma~\ref{prop:intGL-sub}(v)) we obtain that there exist $\Delta_{i}$ and $\tT_{i}, i=0, \ldots, n$ such that $\Gamma', y:\cap_{i}^{n}\tT_{i} \vdashsub \cntx'[M\isub{N}{x}]:\tS$ and for all $i \in \{0, \ldots, n\}$, $\Delta_{i} \vdashsub P:\tT_{i}$ and $\Gamma = \Gamma', \Delta_{0}^{\top} \sqcap \Delta_{1} \sqcap \ldots \sqcap \Delta_{n}.$
Since $M\isub{N}{x} \rclsubredc M'$
by IH we have that $\Gamma', y:\cap_{i}^{n}\tT_{i} \vdashsub
\cntx'[M']:\tS$. Using rule $(Subst)$ we can conclude that $\Gamma \vdashsub
  \cntx'\isub{P}{y}[M']:\tS$.

\end{itemize}
\end{proof}

\begin{lemma}[Substitution lemma for $\rcl\cap$]
\label{prop:sub-lemma} If $\;\Gamma, x:\cap_i^{n} \tT_i \vdash
M:\tS\;$ and  $\;\Delta_i
\vdash N:\tT_i$, for all $i \in \{0, \ldots, n\}$, then $\;\Gamma,\dztop \sqcap \Delta_1
\sqcap ... \sqcap \Delta_n \vdash M\ISUB{N}{x}:\tS.$
\end{lemma}

\begin{proof} \rule{0in}{0in}
From assumptions $\;\Gamma, x:\cap_i^{n} \tT_i \vdash
M:\tS\;$ and  ${\Delta_i
\vdash N:\tT_i}$, for all $i \in \{0, \ldots, n\}$, we get that $\;\Gamma, x:\cap_i^{n} \tT_i \vdashsub
M:\tS\;$ and for all $i \in \{0, \ldots, n\}$, $\;\Delta_i
\vdashsub N:\tT_i$. Applying $(Subst)$ rule we get $\;\Gamma,\dztop \sqcap \Delta_1
\sqcap ... \sqcap \Delta_n \vdashsub M\isub{N}{x}:\tS.$ Now, using termination and confluence of $\rclsubred$ reduction (Proposition~\ref{prop:box-term} and Proposition~\ref{prop:rclsubred-conf}) and preservation of type under the $\rclsubred$ reduction (Lemma~\ref{lem:sub-typpres}) we obtain that the unique normal form $M\ISUB{N}{x}$ exists and that $\;\Gamma,\dztop \sqcap \Delta_1
\sqcap ... \sqcap \Delta_n \vdashsub M\ISUB{N}{x}:\tS.$ Since $M\ISUB{N}{x} \in \Rcl$ (Proposition~\ref{prop:nf}), having that $\rclsub\cap$ is conservative extension of $\rcl\cap$, we finally get that $\;\Gamma,\dztop \sqcap \Delta_1
\sqcap ... \sqcap \Delta_n \vdash M\ISUB{N}{x}:\tS.$
\end{proof}

\begin{proposition}[Type preservation under reduction and equivalence in $\rcl\cap$]
\label{prop:sr} For every $\rcl$-term $M$: if $\;{\Gamma \vdash
M:\tS}\;$ and $M \rightarrowc M'$ or $M\equiv_{\rcl}  M'$, then $\;\Gamma \vdash
M':\tS.$
\end{proposition}
\begin{proof}
The proof is done by case analysis on the applied reduction. Since
 the property is stable by context, we can without loss of
 generality assume that the
 reduction takes place at the outermost position of the term.
 Here we just show several cases. We will use GL as an abbreviation for Generation lemma (Lemma~\ref{prop:intGL}).
  \begin{itemize}
  \item Case $(\beta)$: Let $\Gamma \vdash (\lambda x.M)N:\tS$. We
  want to show that $\Gamma \vdash M\ISUB{N}{x}:\tS$.
  From $\Gamma \vdash (\lambda x.M)N:\tS\;$ and from GL(ii) it follows that
  $\Gamma = \Gamma',\dztop \sqcap \Delta_1 \sqcap \ldots \sqcap \Delta_n$,  and that there is a type $\cap_i^{n} \tT_i$ such that
  for all $i=0, \ldots, n$, $\Delta_i \vdash N:\tT_i,\;$ and
  $\Gamma' \vdash \lambda x.M:\cap_i^n\tT_i \to \tS$. Further,
  by GL(i) we have that $\Gamma',x: \cap_i^n\tT_i \vdash M:\tS$.
  Now, all the assumptions of Substitution lemma~\ref{prop:sub-lemma}
  hold, yielding $\Gamma',\dztop \sqcap \Delta_1 \sqcap \ldots \sqcap \Delta_n \vdash M\ISUB{N}{x}:\tS$
  which is exactly what we need, since $\Gamma = \dztop \sqcap \Gamma',\Delta_1 \sqcap \ldots \sqcap \Delta_n$.
  \item Case $(\gamma\omega_2)$: Let $\Gamma \vdash
 \cont{x}{x_1}{x_2}{\weak{x_1}{M}}:\tS$.
 We are showing that $\Gamma \vdash M\ISUB{x}{x_2}:\tS$.\\
 From the first sequent by GL(iii) we have that $\Gamma =
 \Gamma',x:\tA \cap \tB$ and $\Gamma', x_1:\tA,x_2:\tB
 \vdash\weak{x_1}{M}:\tS$. Further, by GL(iv) we conclude that $\tA \equiv \top$, $x:\top \cap \tB \equiv \tB$ and $\Gamma', x_2:\tB \vdash
 M:\tS$. Since $\tB= \cap_i^n \tT_i$ for some $n \geq 0$, by applying Substitution lemma~\ref{prop:sub-lemma} to  $\Gamma', x_2:\tB \vdash
 M:\tS$ and $x:\tT_i \vdash x:\tT_i,\;i=0,\ldots,n$ we
 get $\Gamma \vdash M\ISUB{x}{x_2}:\tS$.
 \item The other rules are easy since they do not essentially
 change the structure of the term.
\end{itemize}
\end{proof}

Due to this property, equivalent (by $\equiv_{\rcl}$) terms have the same type.

\section{Characterisation of strong normalisation in $\rcl$}
\label{sec:typeSN}

\subsection{SN $\Rightarrow$ Typeability in $\rcl \cap$}
\label{sec:SNtype}

We want to prove that if a $\rcl$-term is strongly normalising (SN), then it is typeable
in the system $\rcl \cap$. We proceed in two steps:
\begin{enumerate}
\item we show that all $\rcl$-normal forms are typeable and
\item we prove the redex subject expansion.
\end{enumerate}

\begin{proposition}\label{prop:nf-are-typ}
$\rcl$-normal forms are typeable in the system $\rcl \cap$.
\end{proposition}
\begin{proof}
By induction on the structure of $M_{nf}$ and $E_{nf}$, given in Definition~\ref{def:setOfNF}. The basic case is a variable, namely $xM_{nf}^{1} \ldots M_{nf}^{n}$, where $n=0$. It is typeable by $(Ax)$. Cases involving duplication and erasure operators are easy, because the associated type assignment rules $(Cont)$ and $(Thin)$ preserve the type of a term. If $M_{nf} = \lambda x.\weak{x}{N_{nf}$}, then by the induction hypothesis $\Gamma \vdash N_{nf}:\sigma$, hence  $\Gamma, x:\top \vdash \weak{x}{N_{nf}}:\sigma$ and $\Gamma \vdash \lambda x.\weak{x}{N_{nf}}:\top \to \sigma$. Further, we discuss the case $xM_{nf}^{1} \ldots M_{nf}^{n}$, where $n \geq 1$. In this case, $M_{nf}^{1}, \ldots ,M_{nf}^{n}$ are typeable by the induction hypothesis, say $\Gamma_{j}^{i} \vdash M_{nf}^{i}:\sigma_{j}^{i},\; i\in\{1,...,n\},\; j\in\{1,...,m_i\}$. Then, since $x$ is a fresh variable, taking $x:\cap_{j}^{m_1}\sigma_{j}^{1} \to (\cap_{j}^{m_2}\sigma_{j}^{2} \to \ldots (\cap_{j}^{m_n}\sigma_{j}^{n} \to \tT)\ldots)$ and applying $(\to_E)$ rule $n$ times, we obtain $\Gamma \vdash xM_{nf}^{1} \ldots M_{nf}^{n}:\tT$, where $\Gamma= x:\cap_{j}^{m_1}\sigma_{j}^{1} \to (\cap_{j}^{m_2}\sigma_{j}^{2} \to \ldots (\cap_{j}^{m_n}\sigma_{j}^{n} \to \tT)\ldots), \gztopone \sqcap \Gamma_{1}^{1}\sqcap ... \sqcap \Gamma_{m_1}^{1}, \ldots, \gztopn \sqcap \Gamma_{1}^{n}\sqcap ... \sqcap \Gamma_{m_n}^{n}$.
\end{proof}




\begin{lemma}
\label{prop:box-exp}
For all $M,M' \in \Rclsub$ and $N \in \Rcl$, if $\;\Gamma \vdashsub M':\tS$, $M\isub{N}{x} \rclsubred M'$,
and $N$ is typeable, then $\Gamma \vdashsub M\isub{N}{x}:\tS$.
\end{lemma}

\begin{proof}
The proof is by case analysis on the applied $\rclsubred$ reduction.
We consider only some interesting rules.
\begin{itemize}
\item Rule $(\weak{x}{M})\isub{N}{x}  \rclsubred   \weak{Fv(N)}{M}$.\\
Let $Fv(N)= \{x_1,...,x_m\}$. By assumption $N$ is typeable, thus  $\Delta_0 \vdash N:\tT_0$ for some $\Delta_0=\{x_1:\tT_1,...,x_m:\tT_m\}$. If $\Gamma \vdashsub \weak{Fv(N)}{M}:\tS$, then by applying $m$ times the Generation Lemma~\ref{prop:intGL-sub}(iv), we get  $\Gamma' \vdashsub M:\tS$, where $\Gamma=\Gamma',\dztop$.
On the other hand
\[
 \prooftree
    \prooftree
    \Gamma' \vdashsub M:\tS
    \justifies `G', x:\top \vdashsub \weak{x}{M}:\tS
    \using (Thin)
    \endprooftree
    `D_0 \vdash N:`t_0
    \justifies `G',\dztop \vdashsub (\weak{x}{M})\isub{N}{x}:\tS.
    \using (Subst)
   \endprooftree
\]
Notice that the rule $(Subst)$ can be applied because $\top=\cap_{i}^{n}\tT_i$ for $n=0$.
\item Rule $(\cont{x}{x_1}{x_2}{M})\isub{N}{x} \rclsubred  \cont{Fv[N]}{Fv[N_1]}{Fv[N_2]}{M\isub{N_1}{x_1}\isub{N_2}{x_2}}$.\\
Let $Fv[N]= [y_1,...,y_m]$. Then, since $N_1$ and $N_2$ are obtained from $N$ by renaming the free variables, we have that $Fv[N_1]= [y'_1,...,y'_m]$ and $Fv[N_2]= [y''_1,...,y''_m]$.
From the assumption $\Gamma \vdashsub \cont{Fv[N]}{Fv[N_1]}{Fv[N_2]}{M\isub{N_1}{x_1}\isub{N_2}{x_2}}:\tS$, by $m$ applications of Lemma~\ref{prop:intGL-sub}$(iii)$, we obtain that $\Gamma = \Gamma',y_1:\tT_1 \cap \tR_1,...,y_m:\tT_m \cap \tR_m$ and that $\Gamma', \Delta', \Delta'' \vdashsub M\isub{N_1}{x_1}\isub{N_2}{x_2}:\tS$, where $\Delta' = \{y'_1:\tT_1,...,y'_m:\tT_m\}$ and $\Delta'' = \{y''_1:\tR_1,...,y''_m:\tR_m\}$.
Now, by two applications of Lemma~\ref{prop:intGL-sub}$(v)$,
we get that $\Delta' = \dztopp \sqcap \Delta_1'...\sqcap \Delta'_{n_1}$, $\Delta'' = \dztoppp \sqcap \Delta_1''...\sqcap \Delta''_{n_2}$, where $\Delta'_i = \{y'_1:\tT_{1,i},...,y'_m:\tT_{m,i}\}$
for $i \in \{0,...,n_1\}$, $\Delta''_j = \{y''_1:\tR_{1,j},...,y''_m:\tR_{m,j}\}$ for $j \in \{0,...,n_2\}$, $\Delta'_i \vdashsub N_1:\cap_k^m \tT_{k,i}$, $\Delta''_j \vdashsub N_2:\cap_k^m \tR_{k,j}$, and finally $\Gamma', x_1:\cap_i^{n_1} \tT_i,x_2:\cap_j^{n_2}\tR_j \vdashsub M:\tS$ (we used the following abbreviations: $\cap_k^m \tT_{k,i} \equiv \tT_i,\,\cap_k^m \tR_{k,j} \equiv \tR_j$).
Now, since $N_1$ and $N_2$ are obtained from $N$ by renaming, for each derivation of the type of $N_1$ (respectively $N_2$) we can write an analogous derivation of the type of $N$, i.e.
$\Delta_i \vdashsub N:\tT_i$ for $i \in \{0,...,n_1\}$ and $\Delta_j \vdashsub N:\tR_j$ for $j \in \{0,...,n_2\}$, where $\Delta_i$ differ from $\Delta'_i$ (and respectively $\Delta_j$ from $\Delta''_j$) only by the domain ($Dom(\Delta_i)=Dom(\Delta_j)=\{y_1,...,y_m\}$). If we adopt abbreviations $\mathfrak{L}_1$ for the array of the first $n_1$ derivations, and $\mathfrak{L}_2$ for the array of the latter $n_2$ derivations, we have:
\[
 \prooftree
    \prooftree
    \Gamma', , x_1:\cap_i^{n_1} \tT_i,x_2:\cap_j^{n_2}\tR_j \vdashsub M:\tS
    \justifies `G', x:(\cap_i^{n_1} \tT_i) \cap (\cap_j^{n_2}\tR_j) \vdashsub \cont{x}{x_1}{x_2}{M}:\tS
    \using (Cont)
    \endprooftree
    \quad \mathfrak{L}_1 \quad \mathfrak{L}_2
    \justifies `G \vdashsub (\cont{x}{x_1}{x_2}{M})\isub{N}{x}:\tS.
    \using (Subst)
   \endprooftree
\]
The left hand side of the latter assignment holds because $`G', \dztop \sqcap\Delta_1 \sqcap ... \sqcap \Delta_{n_1+n_2} = `G', y_1:\top \cap (\cap_i^{n_1} \tT_{1,i}) \cap (\cap_j^{n_2}\tR_{1,j}),..., y_m:\top \cap (\cap_i^{n_1} \tT_{m,i}) \cap (\cap_j^{n_2}\tR_{m,j}) = `G', y_1:\tT_1 \cap \tR_1,...,y_m:\tT_m \cap \tR_m = \Gamma$.
\end{itemize}
\end{proof}

\begin{proposition}[Redex subject expansion]
\label{prop:sub-exp}
\rule{0in}{0in}
\begin{itemize}
\item[(i)]
If $\Gamma \vdash M\ISUB{N}{x}:\tS$ and $N$ is typeable,
then ${\Gamma \vdash (\lambda x.M)N:\tS}$.
\item[(ii)]
Let $M$ be a $\rcl$-redex other than a $\beta$-redex and  $M \to M'$. 
If $\Gamma \vdash M':\tS$, then $\Gamma \vdash M:\tS$.
\end{itemize}
\end{proposition}

%

\begin{proof}
(i) From $\Gamma \vdash M\ISUB{N}{x}:\tS$ we have that $\Gamma \vdashsub M\isub{N}{x}:\tS$ using Lemma~\ref{prop:box-exp} multiple times, since $M\ISUB{N}{x} = M\isub{N}{x} \subnf$, i.e.\  $M\isub{N}{x} \rclsubredc M\ISUB{N}{x}$. From ${\Gamma \vdashsub M\isub{N}{x}:\tS}$ by Lemma~\ref{prop:intGL-sub}(v) (Generation lemma) it follows that there exist $`D_i$ and $\tT_i,\;i  = 0, \ldots, n$ such that
$\Gamma', x:\cap_{i}^{n}\tT_i \vdashsub M:\tS$ and for all $i \in
\{0, \ldots, n\}$, $\;\Delta_{i} \vdash N:\tT_i$ and $\;\Gamma=\Gamma',
\dztop \sqcap \Delta_{1} \sqcap \ldots \sqcap \Delta_{n}$. Now:
\[
 \prooftree
    \prooftree
    \Gamma', x:\cap_{i}^{n}\tT_i \vdashsub M:\tS
    \justifies `G' \vdashsub \lambda x.M:\cap_{i}^{n}\tT_i \to \tS
    \using (\to_I)
    \endprooftree
    `D_0 \vdashsub N:`t_0 \quad ... \quad `D_n \vdashsub N:`t_n
    \justifies `G \vdashsub (\lambda x.M)N:\tS
    \using (\to_E)
    \endprooftree
\]
Since $M,N \in \Rcl$ we have that $`G \vdash (\lambda x.M)N:\tS$.

(ii)
By case analysis according to the applied reduction, similar to the proof of Proposition~\ref{prop:sr}.
\end{proof}

\begin{theorem}[SN $\Rightarrow$ typeability]\label{thm:SNtypable}
All strongly normalising $\rcl$-terms are typeable in the
$\rcl\cap$ system.
\end{theorem}

\begin{proof}
The proof is by induction on the length of the longest reduction
path out of a strongly normalising term $M$, with a subinduction
on the structure of $M$.

\begin{itemize}

\item If $M$ is a normal form, then $M$ is typeable by
  Proposition~\ref{prop:nf-are-typ}.

\item If $M$ is a $\rcl$-redex, i.e.\ $M \to M'$, then let $M'$ be
  its contractum. $M'$ is also strongly normalising, hence by
  IH it is typeable. Then $M$ is typeable, by
  Proposition~\ref{prop:sub-exp}. Notice that, if $M \equiv (\lambda x.N)P
  \to_{\beta} N\ISUB{P}{x} \equiv M'$, then, by IH, $P$ is typeable, since the length of
  the longest reduction path out of $P$ is smaller than that of $M$.

\item
Next, suppose that $M$ itself is neither a redex nor a normal form.
Then, according to Lemma~\ref{lem:hft}, $M$ has of one of the following forms:
\begin{enumerate}
\item[-] $\lambda x.N$ (where $N \not = \weak{y}{P}$ and $y\not = x$, since in this case $M$ would be a redex and previous case would apply),
\item[-] $xT_{1} \ldots T_{n}$,
\item[-] $\weak{x}{N}$,
\item[-] $(\lambda x.N)PT_{1} \ldots T_{n}$,
\item[-] $(\weak{x}{N})PT_{1} \ldots T_{n}$,
\item[-] $(\cont{x}{x_{1}}{x_{2}}{N})T_{1} \ldots T_{n}$,
\end{enumerate}
where $N, P, T_{1}, \ldots, T_{n}$, are \emph{not} all normal forms.
We can classify these forms into the following two categories:

\begin{enumerate}

\item[1)] Terms with internal redexes: $\lambda x.N$, $xT_{1} \ldots T_{n}$, $\weak{x}{N}$ and $(\cont{x}{x_{1}}{x_{2}}{N})T_{1} \ldots T_{n}$ when duplication cannot be propagated further into $N$, i.e. $N\equiv PQ,\,x_1 \in Fv(P),\,x_2 \in Fv(Q)$. In all these cases, we proceed by subinduction on the structure of $M$, since the length of the longest reduction path out of a subterm that contains a redex is equal to the length of the longest reduction path out of $M$. 

\item[2)] Terms with a leftmost redex: 
$(\lambda x.N)PT_{1} \ldots T_{n}$, $(\weak{x}{N})PT_{1} \ldots T_{n}$ and $(\cont{x}{x_{1}}{x_{2}}{N})T_{1} \ldots T_{n}$ when duplication can be propagated further into $N$. In these cases, by applying the leftmost reduction, we obtain a term with smaller length of the longest reduction path, therefore we can proceed using induction.

\end{enumerate}

In all the cases, after the application of induction (respectively subinduction) hypothesis in order to conclude typeability of subterms of $M$, it is easy to build the type of $M$.
We will prove some illustrative cases from both categories, the rest being similar.
\begin{itemize}

\item 
$M \equiv \lambda x.N$. Then, the only way to reduce $M$ is to reduce $N$ and the number of reductions in $N$ is equal to the number of reductions in $M$. Since $M$ is SN, $N$ is also $SN$. 
Since $N$ is a subterm of $M$, $N$ is typeable by subinduction and $\lambda x.N$ is typeable by $(\to_I)$.

\item
$M \equiv xT_{1} \ldots T_{n}$. Then $T_{1}, \ldots, T_{n}$ must be SN by subinduction, hence typeable. Then we build the type for $M$ by multiple application of the rule $(\to_{E})$, as in Proposition~\ref{prop:nf-are-typ}.

\item
$M \equiv (\cont{x}{x_1}{x_2}{PQ})T_{1} \ldots T_{n}$ with $x_1 \in Fv(P),\;x_2 \in Fv(Q)$. Again, each of $P, Q, T_{1}, \ldots, T_{n}$ must be SN by subinduction, hence typeable. We first use the rule $(Cont)$ to type $\cont{x}{x_1}{x_2}{PQ}$ and then we use the rule $(\to_{E})$, as in Proposition~\ref{prop:nf-are-typ} to type $M$.

\item
$M \equiv (\lambda x.N)PT_{1} \ldots T_{n}$. Then
$M \to M'$ where $M' \equiv N\ISUB{P}{x}T_{1} \ldots T_{n}$. $M'$ is also SN, hence typeable by induction hypothesis, since the longest reduction path out of $M'$ is smaller than the one out of $M$. 
This implies that $N\ISUB{P}{x}, T_{1}, \ldots, T_{n}$ are also SN and hence typeable by sub induction. Then we build the type for $M$ by multiple application of the rule $(\to_{E})$, as in Proposition~\ref{prop:nf-are-typ}. The cases $M \equiv (\weak{x}{N})PT_{1} \ldots T_{n}$ and $M \equiv (\cont{x}{x_{1}}{x_{2}}{N})T_{1} \ldots T_{n}$ are analogous.

\end{itemize}
\end{itemize}
\end{proof}

\subsection{Typeability $\Rightarrow$ SN in $\rcl \cap$}
\label{sec:reducibility}

In various type assignment systems, the \emph{reducibility method} can be used to prove many reduction properties of typeable terms.
It was first introduced by Tait~\cite{tait67} for proving the strong normalisation
of simply typed $\lambda$-calculus, and developed further to prove \emph{strong normalisation} of various calculi in~\cite{tait75,gira71,kriv90,ghil96,ghillika02}, \emph{confluence} (the Church-Rosser property)
of $\beta \eta$-reduction in~\cite{kole85,statI85,mitc90,mitc96,ghillika02}
and to characterise certain classes of $\lambda$-terms
such as strongly normalising, normalising, head normalising,
and weak head normalising terms (and their persistent versions) by their typeability in various intersection type systems in~\cite{gall98,dezahonsmoto00,dezaghil02,dezaghillika04}.

The main idea of the reducibility method is to interpret types by suitable sets of lambda terms which satisfy some realisability properties
and prove the soundness of type assignment with respect to these interpretations.
A consequence of soundness is that every typeable term belongs to the interpretation of its type, hence satisfying a desired reduction property.

In the sequel, we adapt the reducibility method in order to prove that terms typeable in $\rcl \cap$ are strongly normalising.

%
%
%
%

\begin{definition}
\label{def:fsto}
For $\vM, \vN \subseteq \LR$, we define $\vM \fsto \vN \subseteq
\LR$ as
$$\vM \fsto \vN =  \{M  \in \LR \mid \forall N \in \vM \quad MN \in \vN \}.$$
\end{definition}

\begin{definition}

\label{def:typeInt} The type interpretation  $\ti{-} : \mathsf{Types} \to
2^{\LR}$ is defined by:
\begin{itemize}
    \item[($I 1$)] $\ti{p} = \SN$, where $p$ is a type atom;
    \item[($I 2$)] $\ti{\tA \to \tS} = \ti{\tA} \fsto \ti{\tS}$;
    \item[($I 3$)]
    $\ti{\cap^n_i \tS_i}  = \left\{ \begin{array}{rr}
 \cap^n_i \ti{\tS_i} & \mbox {for } n > 0\\
 \SN & \mbox{ for } n=0.
\end{array}
\right.$
\end{itemize}

\end{definition}


Next, we introduce the notions of \emph{variable property},
$\beta$-\emph{expansion property}, $\omega$-\emph{expansion property}, $\gamma$-\emph{reduction property},
\emph{thinning property} and \emph{contraction property.}
The variable property and the $\beta$-expansion property correspond to the saturation property given in~\cite{bare92}.

\begin{definition}\label{def:var+sat}
\rule{0in}{0in}
\begin{itemize}

\item A set $\vX \subseteq \LR$ satisfies the \emph{variable property}, notation
$\VAR(\vX)$, if  $\vX$ contains all the terms of the form $x M_1 \ldots M_n$, where $n \geq 0$ and  $M_i`:
\SN$, $i=1,\ldots,n$.


\item A set $\vX \subseteq \LR$ satisfies the \emph{$`b$-expansion property}, notation
  $\SAT_{`b}(\vX)$ if
  \begin{displaymath}
    \prooftree
M_1\in \SN \;\ldots \;M_n \in \SN\;\; N\in\SN \qquad M\ISUB{N}{x} M_1\ldots M_n\in \vX
        \using \SAT_{`b}(\vX) %
\justifies  (`l x. M) \, N \, M_1\ldots M_n\in \vX.
\endprooftree
  \end{displaymath}

\item A set $\vX \subseteq \LR$ satisfies the \emph{$\omega$-expansion property}, notation
  $\SAT_{\omega}(\vX)$ if
  \begin{displaymath}
    \prooftree
M_1\in \SN \;\ldots \; M_n \in \SN\;\; N\in\SN \qquad \weak{x}{(MN)}M_1\ldots M_n\in \vX
        \using \SAT_{\omega}(\vX) %
\justifies  (\weak{x}{M}) \, N \, M_1\ldots M_n\in \vX.
\endprooftree
  \end{displaymath}
\item A set $\vX \subseteq \LR$ satisfies the \emph{$\gamma$-reduction property}, notation
  $\RED_{\gamma}(\vX)$ if
  \begin{displaymath}
    \prooftree
M_1\in \SN \;\ldots\; M_n \in \SN\;\; N\in\SN \qquad \cont{x}{x_{1}}{x_{2}}{(MN)}M_1\ldots M_n\in \vX
        \using \RED_{\gamma}(\vX) %
\justifies  (\cont{x}{x_{1}}{x_{2}}{M}) \, N \, M_1\ldots M_n\in \vX.
\endprooftree
  \end{displaymath}

\item A set $\vX \subseteq \LR$ satisfies the \emph{thinning property}, notation $\WEAK(\vX)$ if:
    \begin{center}
      \prooftree M \in \vX %
      \using \WEAK(\vX) %
      \justifies \weak{x}{M} \in \vX.
      \endprooftree
    \end{center}
\item A set $\mathcal{X} \subseteq \LR$ satisfies the \emph{contraction property}, notation $\CONT(\vX)$ if:
      \begin{center}
        \prooftree %
        M \in \vX %
        \using \CONT(\vX) %
        \justifies \cont{x}{y}{z}{M} \in \vX.
        \endprooftree
      \end{center}
  \end{itemize}
\end{definition}

\noindent\textbf{Remark.} In Definition~\ref{def:var+sat} it is not necessary to explicitly write the conditions about free variables since we work with  $\rcl$-terms.

\begin{definition}[$\circledR$-Saturated set] \label{lem:saturSet}
A set $\vX \subseteq \LR$ is called \emph{$\circledR$-saturated},
if $\vX\subseteq \SN$ and $\vX$ satisfies the variable, $\beta$-expansion,
$\omega$-expansion, $\gamma$-reduction,
thinning and contraction properties.
\end{definition}




\begin{proposition} \label{prop:saturSets}
Let $\vM, \vN \subseteq \LR$.
\begin{itemize}
    \item[(i)] $\SN$ is $\circledR$-saturated.
    \item[(ii)] If $\vM$ and $\vN$ are $\circledR$-saturated, then $\vM \fsto \vN$ is $\circledR$-saturated.
    \item[(iii)] If $\vM$ and $\vN$ are $\circledR$-saturated, then $\vM \cap \vN$ is $\circledR$-saturated.
    \item[(iv)] For all types $\varphi \in \tlam$, $\ti{\varphi}$ is $\circledR$-saturated.
\end{itemize}
\end{proposition}


\begin{proof}
\rule{0in}{0in}
%
(i)
\begin{itemize}
\item
$\SN \subseteq \SN$ and  $\VAR(\SN)$ 
trivially hold.

\item
$\SAT_{\beta}(\SN)$.
Suppose that
	$M\ISUB{N}{x} M_1 \ldots M_n \in \SN$, $M_1, \ldots, M_n \in \SN$ and $N\in\SN$.
We know that $M\ISUB{N}{x} \in \SN$ as a subterm of a term in $\SN$ and $N \in \SN$, hence $M \in \SN$.  By assumption, $M_{1}, \ldots, M_{n} \in \SN$, so all reductions inside of these terms terminate. Starting from $(\lambda x.M)N M_1 \ldots M_n$, we can either contract the head redex and obtain $M\ISUB{N}{x} M_1 \ldots M_n$ which is SN by assumption, so we are done, or we can contract redexes inside $M, N, M_1, $ $\ldots, M_n$, which are all SN by assumption. All these reduction paths are finite. Consider a term obtained
after finitely many reduction steps
	$$(\lambda x.M)N M_1 \ldots M_n \rightarrow \ldots \rightarrow (\lambda x.M')N' M'_1 \ldots M'_n$$
where $M \rightarrowc M',\; N \rightarrowc N', \; M_{1} \rightarrowc M'_{1}, \ldots, M_{n} \rightarrowc M'_{n}.$ After contracting the head redex of $(\lambda x.M')N' M'_1 \ldots M'_n$ to $M'\ISUB{N'}{x} M'_1 \ldots M'_n$, we actually obtain a reduct of $M\ISUB{N}{x}M_1 \ldots M_n \in \SN$. Hence, $(\lambda x.M)N M_1 \ldots M_n \in \SN. $

\item
$\SAT_{\omega}(\SN)$. 
Suppose that
	$\weak{x}{(MN)} M_1 \ldots M_n \in \SN$, $M_1, \ldots, M_n \in \SN$.
Since $\weak{x}{(MN)}$ is a subterm of a term in $\SN$, we know that $MN \in \SN$ and consequently $M,N \in \SN$. By assumption, $M_{1}, \ldots, M_{n} \in \SN$, so the reductions inside of these terms terminate. Starting from $(\weak{x}{M})NM_{1}\ldots M_{n}$, we can either contract the head redex and obtain $\weak{x}{(MN)} M_1 \ldots M_n$ which is SN by assumption, so we are done, or we can contract redexes inside $M,N,M_{1}, \ldots, M_{n}$, which are all SN by assumption. All these reduction paths are finite. Consider a term obtained after finitely many reduction steps
	$$(\weak{x}{M})NM_{1}\ldots M_{n} \rightarrow \ldots \rightarrow (\weak{x}{M'})N'M'_{1}\ldots M'_{n}$$
where $M \rightarrowc M',\; M_{1} \rightarrowc M'_{1}, \ldots, M_{n} \rightarrowc M'_{n}. $ After contracting the head redex of $(\weak{x}{M'})N'M'_{1}\ldots M'_{n}$ to $\weak{x}{(M'N'})M'_{1}\ldots M'_{n}$, we obtain a reduct of $\weak{x}{(MN)} M_1 \ldots M_n \in \SN$. Hence, $(\weak{x}{M})NM_{1}\ldots M_{n} \in \SN. $

\item $\RED_{\gamma}(\SN)$. This is trivial, since by reducing a SN term we again obtain a SN term.

\item $\WEAK(\SN)$.
Suppose that
$M \in \SN$ and $x \not \in Fv(M)$. Then trivially $\weak{x}{M} \in \SN$, since no new redexes are formed.
\item
$\CONT(\SN)$.
Suppose that $M \in \SN,\; y \not = z,\; y, z \in Fv(M),\; x \not \in Fv(M) \setminus \{y,z\}$.
We prove that $\cont{x}{y}{z}{M} \in \SN$ by induction on the structure of $M$.
\begin{itemize}
\item $M = yz$. Then $\cont{x}{y}{z}{M} = \cont{x}{y}{z}{(yz)}$ which is a normal form.
\item $M = \weak{y}{z}$. Then $\cont{x}{y}{z}{M} = \cont{x}{y}{z}{(\weak{y}{z})} \rightarrow_{\gamma\omega_{2}} z\ISUB{x}{z}=x \in \SN$.
\item $M = \lambda w.N$. Then $N \in \SN$ and $\cont{x}{y}{z}{M} = \cont{x}{y}{z}{(\lambda w.N)} \rightarrow_{\gamma_{1}} \lambda w.\cont{x}{y}{z}{N} \in \SN$, since $\cont{x}{y}{z}{N} \in \SN$ by IH.
\item $M = PQ$. Then $P, Q \in \SN$ and if $y,z \not \in Fv(Q)$, $\cont{x}{y}{z}{M} = \cont{x}{y}{z}{(PQ)} \rightarrow_{\gamma_{2}} (\cont{x}{y}{z}{P})Q \in \SN$, since by IH $\cont{x}{y}{z}{P} \in \SN$.\\ The case of $\rightarrow_{\gamma_{3}}$ reduction when $y,z \not \in Fv(P)$ is analogous.
\item $M = \weak{w}{N}$. Then $\cont{x}{y}{z}{M} = \cont{x}{y}{z}{(\weak{w}{N})} \rightarrow_{\gamma\omega_{1}} \weak{w}{(\cont{x}{y}{z}{N})}$. By IH $\cont{x}{y}{z}{N} \in \SN$ and $\weak{w}(\cont{x}{y}{z}{N})$ does not introduce any new redexes.
\item $M = \weak{y}{N}$. Then $\cont{x}{y}{z}{M} = \cont{x}{y}{z}{(\weak{y}{N})} \rightarrow_{\gamma\omega_{2}} N\ISUB{x}{z} \in \SN$, since $N \in \SN$ by IH.
\item $M=\cont{y}{u}{v}{N}$. Then the only possible reduction is inside the term $N$ which is strongly normalising as a subterm of the strongly normalising term $M=\cont{y}{u}{v}{N}$.
\item $M=\cont{x_{1}}{y_{1}}{z_{1}}{N}$. Analogous to the previous case.
\end{itemize}
\end{itemize}

(ii)
  \begin{itemize}
  \item $\vM \fsto \vN \subseteq \SN$.  Suppose that $M \in \vM \fsto
    \vN$. 
    Then, for all $N \in \vM,\; MN \in \vN$. Since $\vM$ is $\circledR$-saturated,
    $\VAR(\vM)$ holds so $x \in \vM$ and $Mx \in \vN \subseteq \SN.$ From here we can
    deduce that $M \in \SN$.

  \item $\VAR(\vM \fsto \vN)$.  Suppose that $x$ is a variable and $M_1, \ldots, M_n \in \SN,
    n \geq 0$, such that $x \cap Fv(M_1) \cap \ldots \cap Fv(M_n) = \emptyset$. We need to show
    that $x M_1 \ldots M_n \in \vM \fsto \vN,$ i.e.\ $\forall N \in \vM$, $x M_1 \ldots
    M_nN \in \vN$. This holds since by assumption $\vM \subseteq \SN$ and $\vN$ is
    $\circledR$-saturated, i.e.\ $\VAR(\vN)$ holds.

   \item $\SAT_{\beta}(\vM \fsto \vN)$.  Suppose that $M\ISUB{N}{x} M_1 \ldots M_n \in \vM
    \fsto \vN$, $M_1, \ldots,$ $M_n \in \SN$ and $N \in \SN$. This means that for all $P \in \vM$,
    $M\ISUB{N}{x} M_1 \ldots M_nP \in \vN.$ But $\vN$ is $\circledR$-saturated, so
    $\SAT_{\beta}(\vN)$ holds and we have that for all $P \in \vN$, $(\lambda x.M)N M_1
    \ldots M_nP \in \vN.$ This means that $(\lambda x.M)N M_1 \ldots M_n \in \vM \fsto \vN. $

\item $\SAT_{\omega}(\vM \fsto \vN)$. Analogous to $\SAT_{\beta}(\vM \fsto \vN)$.

  \item $\RED_{\gamma}(\vM \fsto \vN)$. Suppose that $\cont{x}{x_{1}}{x_{2}}{(MN)} \in \vM \fsto \vN$.
  This means that for all $P \in \vM, \cont{x}{x_{1}}{x_{2}}{(MN)}P \in \vN$. But $\vN$ is
    $\circledR$-saturated, i.e.\ $\RED_{\gamma}(\vN)$ holds, hence $(\cont{x}{x_{1}}{x_{2}}{M})NP \in \vN$.
    This means that $(\cont{x}{x_{1}}{x_{2}}{M})N \in \vM \fsto \vN$.

  \item $\WEAK(\vM \fsto \vN)$.  Suppose that $M \in \vM \fsto \vN$ and $x \not \in
    Fv(M)$. This means that for all $N \in \vM, MN \in \vN$. But $\vN$ is
    $\circledR$-saturated, i.e.\ $\WEAK(\vN)$ holds, hence $\weak{x}{(MN)} \in \vN$. Also
    $\SAT_{\omega}(\vN)$ holds so we obtain for all $N \in \vM, (\weak{x}{M})N \in
    \vN$, i.e.\ $\weak{x}{M} \in \vM \fsto \vN$.
  \item $\CONT(\vM \fsto \vN)$.  Let $M \in \vM \fsto \vN$. We want to prove that
      ${\cont{x}{y}{z} {M}} \in \vM \fsto \vN$ for $y \not = z,\; y, z \in
      Fv(M)$ and ${x \not \in Fv(M)}$. Let $P$ be any term in $\vM$. We have to prove that
      $({\cont{x}{y}{z} {M}}) \, P \in \vN$.  Since $M \in \vM \fsto \vN$, we know that
      $M\,P\in\vN$. By assumption $\vN$ is $\circledR$-saturated so ${\cont{x}{y}{z}{(M\,P)} \in \vN}$. Using
      $\RED_{\gamma}(\vN)$ we obtain $({\cont{x}{y}{z} {M}}) \, P \in \vN$.
    Therefore ${\cont{x}{y}{z} {M}} \in \vM \fsto \vN$.
  \end{itemize}

(iii)
\begin{itemize}
\item
$\vM \cap \vN \subseteq \SN$ is straightforward, since $\vM, \vN \subseteq \SN$ by assumption.
\item
$\VAR(\vM \cap \vN)$. Since $\VAR(\vM)$ and $\VAR(\vN)$ hold, we have that
$\forall M_{1}, \ldots,$ $M_{n}$ $\in \SN$, $n \geq 0$:
$xM_1 \ldots M_{n} \in \vM$ and $xM_1 \ldots M_{n} \in \vN$. We deduce that
$\forall M_{1}, \ldots, M_{n} \in \SN$, $n \geq 0$:
$xM_1 \ldots M_{n} \in \vM \cap \vN$, i.e.\ $\VAR(\vM \cap \vN)$ holds.
\item $\SAT_{\beta}(\vM \cap \vN)$ is straightforward.
\item $\SAT_{\omega}(\vM \cap \vN)$ is straightforward.
\item
$\RED_{\gamma}(\vM \cap \vN)$. Suppose that $\cont{x}{x_{1}}{x_{2}}{(MN)} \in \vM \cap \vN$.
Since both $\vM$ and $\vN$ are $\circledR$-saturated $\RED_{\gamma}(\vM)$ and $\RED_{\gamma}(\vN)$ hold, hence
$(\cont{x}{x_{1}}{x_{2}}{M})N \in \vM$ and  $(\cont{x}{x_{1}}{x_{2}}{M})N \in \vM$, i.e.\ $(\cont{x}{x_{1}}{x_{2}}{M})N \in \vM \cap \vN$.

\item
$\WEAK(\vM \cap \vN)$.
Let $M \in \vM \cap \vN$ and $x \not \in Fv(M)$. Then $M \in \vM$ and $M \in \vN$. Since both $\vM$ and $\vN$ are $\circledR$-saturated $\WEAK(\vM)$ and $\WEAK(\vN)$ hold, hence $\weak{x}{M} \in \vM$ and $\weak{x}{M} \in \vN$, i.e.\ $\weak{x}{M} \in \vM \cap \vN$.
\item
$\CONT(\vM \cap \vN)$.
Suppose that $M \in \vM \cap \vN,\; y \not = z,\; y, z \in Fv(M),\; x \not \in Fv(M) \setminus \{y,z\}$. Since both $\vM$ and $\vN$ are $\circledR$-saturated $\CONT(\vM)$ and $\CONT(\vN)$ hold, hence $\cont{x}{y}{z}{M} \in \vM$ and $\cont{x}{y}{z}{M} \in \vN$, i.e.\ $\cont{x}{y}{z}{M} \in \vM \cap \vN$.
\end{itemize}

(iv)
By induction on the construction of $\varphi \in \mathsf{Types}$.
\begin{itemize}
    \item If $\varphi \equiv p$, $p$ a type atom, then $\ti{\varphi} = \SN$, so it is $\circledR$-saturated using (i).
    \item If $\varphi \equiv \tA \to \tS$, then $\ti{\varphi} = \ti{\tA} \fsto \ti{\tS}$. Since $\ti{\tA}$ and $\ti{\tS}$ are $\circledR$-saturated by assumption, we can use (ii).
    \item If $\varphi \equiv \cap_{i}^{n} \tS_{i}$, then we distinguish two cases:
    \begin{itemize}
    \item for $n>0$, $\ti{\varphi} =\ti{\cap_{i}^{n} \tS_{i}}  = \cap_{i}^{n} \ti{\tS_{i}}$ and for all $i=1, \ldots, n, \ti{\tS_{i}}$ are $\circledR$-saturated by assumption, so we can use (iii).
    \item for $n=0$, $\varphi \equiv \cap_{i}^{0} \tS_{i}$, then $\ti{\varphi} = \SN$ and we can use (i).
    \end{itemize}
\end{itemize}

\end{proof}


We further define a {\em valuation of terms\/} $\tei{-}_{\rho}: \LR \to \LR$ and the {\em semantic satisfiability relation\/}
$\models$ connecting the type interpretation with the term valuation.

\begin{definition}  \label{def:val}
Let $\rho : {\tt var} \to \LR$ be a valuation of term variables in
$\LR$. For ${M \in \LR}$, with $Fv(M) = \{x_1, \ldots, x_n\} $ the
\emph{term valuation} $\tei{-}_\rho : \LR \to \LR$ is defined as follows:
$$\tei{M}_{\rho} = M\ISUBM{\rho(x_1)}{x_1}{\rho(x_n)}{x_n}$$
\noindent providing that $x\not = y \; \Rightarrow \; Fv(\rho(x)) \cap Fv(\rho(y)) = \emptyset$.
\end{definition}

\noindent \textit{Notation:} $`r(N/x)$ is the valuation defined as:
$`r(N/x)(y) = `r(y)$ and ${`r(N/x)(x) = N}$ for $x \not= y$.

\begin{lemma}
\label{lemma:val}
\rule{0in}{0in}
 \begin{itemize}
    \item[(i)]
    $\tei{x}_\rho = \rho(x)$;
    \item[(ii)]
    $\tei{MN}_{\rho} = \tei{M}_{\rho}\tei{N}_{\rho}$;
    \item[(iii)]
    $\tei{\lambda x. M}_{\rho} N  \to_{\beta} \tei{M}_{\rho}\ISUB{N}{x} $ and
    $\tei{M}_{\rho}\ISUB{N}{x} = \tei{M}_{\rho(\isubs{N}{x})}$;
    \item[(iv)]
    $\tei{\weak{x}{M}}_{\rho} = \weak{Fv(\rho(x))}{\tei{M}}_{\rho}$;
    \item[(v)]
    $\tei{\cont{z}{x}{y}{M}}_{\rho} =
    \cont{Fv[N]}{Fv[N_{1}]}{Fv[N_{2}]}{\tei{M}_{\rho(\isubs{N_{1}}{x},\isubs{N_{2}}{y})}}$\\
    where $N=\rho (z)$ and $N_{1}$, $N_{2}$ are obtained from $N$ by renaming its free variables.
\end{itemize}
\end{lemma}

\begin{proof}
\rule{0in}{0in}
 \begin{itemize}

    \item[(i)]
    $\tei{x}_{\rho} = x\ISUB{\rho(x)}{x} = x\isub{\rho(x)}{x}\subnf = \rho(x)$, since $x\isub{\rho(x)}{x} \rclsubred \rho(x).$

    \item[(ii)]
Without loss of generality, we can assume that $Fv(M) = \{x_{1}, \ldots, x_{i}\}$ and $Fv(N) = \{x_{i+1}, \ldots, x_{n}\}$. Then\\
$\tei{MN}_{\rho} = (MN)\ISUBM{\rho(x_{1})}{x_{1}}{\rho(x_{n})}{x_{n}} = \\$
     $M\ISUBM{\rho(x_{1})}{x_{1}}{\rho(x_{i})}{x_{i}} N\ISUBM{\rho(x_{i+1})}{x_{i+1}}{\rho(x_{n})}{x_{n}} = \tei{M}_{\rho}\tei{N}_{\rho}$.

    \item[(iii)]
	If  $Fv(\lambda x.M) = \{x_{1}, \ldots, x_{n}\}$, then\\
	$\tei{\lambda x.M}_{\rho} N =
	  (\lambda x.M)\ISUBM{\rho(x_{1})}{x_{1}}{\rho(x_{n})}{x_{n}} N = \\
	  (\lambda x.M\ISUBM{\rho(x_{1})}{x_{1}}{\rho(x_{n})}{x_{n}} )N \to
	  (M\ISUBM{\rho(x_{1})}{x_{1}}{\rho(x_{n})}{x_{n}} )\ISUB{N}{x}  =
	 \tei{M}_{\rho}\ISUB{N}{x}$.

$\tei{M}_{\rho(\isubs{N}{x})} = 	
M\ISUBM{\rho(\isubs{N}{x})(x_{1})}{x_{1}}{\rho(\isubs{N}{x})(x_{n})}{x_{n},\rho(\isubs{N}{x})(x)\dblsl x} = \\
M\ISUBM{\rho(x_{1})}{x_{1}}{\rho(x_{n})}{x_{n}}\ISUB{N}{x} =
\tei{M}_{\rho}\ISUB{N}{x}$.

    \item[(iv)]
	If  $Fv(M) = \{x_{1}, \ldots, x_{n}\}$, then $Fv(\weak{x}{M}) = \{x, x_{1}, \ldots, x_{n}\}$ and \\
	$\tei{\weak{x}{M}}_{\rho} =
	(\weak{x}{M})\ISUBM{\rho(x) \dblsl x, \rho(x_{1})}{x_{1}}{\rho(x_{n})}{x_{n}} = \\
         \weak{Fv(\rho(x))}{M}\ISUBM{\rho(x_{1})}{x_{1}}{\rho(x_{n})}{x_{n}} =
	\weak{Fv(\rho(x))}{\tei{M}_{\rho}}$ since\\
$  (\weak{x}{M})\isub{\rho(x)}{x}\isub{\rho(x_{1})}{x_{1}} \ldots \isub{\rho(x_{n})}{x_{n}} \rclsubred \\ (\weak{Fv(\rho(x))}{M})\isub{\rho(x_{1})}{x_{1}} \ldots \isub{\rho(x_{n})}{x_{n}} \rclsubredc \\
   \weak{Fv(\rho(x))}{M}\isub{\rho(x_{1})}{x_{1}} \ldots \isub{\rho(x_{n})}{x_{n}} .$

    \item[(v)]
	If  $Fv(M) = \{x,y,x_{1}, \ldots, x_{n}\}$, then $Fv(\cont{z}{x}{y}{M}) = \{z,x_{1}, \ldots, x_{n}\}$ and \\
	$\tei{\cont{z}{x}{y}{M}}_{\rho} =
	(\cont{z}{x}{y}{M})\ISUBM{\rho(z) \dblsl z,\rho(x_{1})}{x_{1}}{\rho(x_{n})}{x_{n}} =  \\
	(\cont{z}{x}{y}{M})\ISUB{N}{z}\ISUBM{\rho(x_{1})}{x_{1}}{\rho(x_{n})}{x_{n}} =  \\
   \cont{Fv[N]}{Fv[N_{1}]}{Fv[N_{2}]}{M\ISUB{N_{1}}{x}\ISUB{N_{2}}{y} \ISUBM{\rho(x_{1})}{x_{1}}{\rho(x_{n})}{x_{n}}}$\\since
$
	(\cont{z}{x}{y}{M})\isub{N}{z}\isub{\rho(x_{1})}{x_{1}} \ldots \isub{\rho(x_{n})}{x_{n}} \rclsubred  \\
	(\cont{Fv[N]}{Fv[N_{1}]}{Fv[N_{2}]}{M}\isub{N_{1}}{x}\isub{N_{2}}{y}) \isub{\rho(x_{1})}{x_{1}} \ldots \isub{\rho(x_{n})}{x_{n}} \rclsubred \\	 
   \cont{Fv[N]}{Fv[N_{1}]}{Fv[N_{2}]}{M\isub{N_{1}}{x}\isub{N_{2}}{y} \isub{\rho(x_{1})}{x_{1}} \ldots \isub{\rho(x_{n})}{x_{n}}} 	.$
	
On the other hand, denoting by $\rho'=\rho(\isubs{N_{1}}{x},\isubs{N_{2}}{y})$ we obtain

$\cont{Fv[N]}{Fv[N_{1}]}{Fv[N_{2}]}{\tei{M}_{\rho(\isubs{N_1}{x},\isubs{N_2}{y})}} = \\
\cont{Fv[N]}{Fv[N_{1}]}{Fv[N_{2}]}
{M\ISUBM{\rho'(x) \dblsl x, \rho'(y) \dblsl y, \rho'(x_{1})}{x_{1}}{\rho'(x_{n})}{x_{n}}} = \\
\cont{Fv[N]}{Fv[N_{1}]}{Fv[N_{2}]}
{M\ISUBM{N_{1} \dblsl x, N_{2} \dblsl y, \rho(x_{1})}{x_{1}}{\rho(x_{n})}{x_{n}}}$


\end{itemize}
\end{proof}

\begin{definition} \label{def:model}
\rule{0in}{0in}
\begin{itemize}
    \item [(i)] $\rho \models M : \tS \quad \iff\ \quad \tei{M}_\rho \in \ti{\tS}$;
    \item [(ii)] $\rho \models \Gamma \quad \iff\ \quad (\forall (x:\tA) \in \Gamma) \quad \rho(x)\in \ti{\tA}$;
    \item [(iii)] $\Gamma \models M : \tS \quad \iff\ \quad (\forall \rho, \rho \models \Gamma \Rightarrow
    		\rho \models M : \tS)$.
  \end{itemize}
\end{definition}

\begin{lemma}\label{lem:valcup}
  Let $`G "|=" M:\tS$ and $`D "|="
  M:\tT$, then $$`r "|=" `G \sqcap `D \mbox{ if and only if } `r "|=" `G \mbox{ and } `r "|=" `D.$$
\end{lemma}

\begin{proof}
  The proof is a straightforward consequence of the definition of bases intersection $\sqcap$.
\end{proof}

\begin{proposition} [Soundness of $\rcl \cap$] \label{prop:sound}
If $\Gamma \vdash M:\tS$, then $\Gamma \models M:\tS$.
\end{proposition}

\begin{proof}
By induction on the derivation of $\Gamma \vdash M:\tS$.

\begin{itemize}

\item
The last rule applied is $(Ax)$, i.e.\
$$\infer[(Ax)]{x:\tS \vdash x:\tS}{}$$
	We have to prove $x:\tS \models x:\tS$.
i.e.\ $(`A`r)\; `r(x) `: \ti{`s} "=>" \tei{x}_\rho`:\ti{`s}$.
This is trivial since according to Lemma~\ref{lemma:val}(i) $\tei{x}_\rho = \rho(x)$.

\item
The last rule applied is $(\to_{I})$, i.e.\
$$\infer[(\to_I)]{\Gamma \vdash \lambda x.M:\tA \to \tS}
                        {\Gamma,x:\tA \vdash M:\tS}$$
By the IH $\Gamma, x:\tA \models M: \tS$ (*).
Suppose that $\rho \models \Gamma$ and we want to show that $\rho \models \lambda x.M: \tA \to \tS$.
We have to show that
$$\tei{\lambda x.M}_{\rho} \in \ti{\tA \to \tS} = \ti{\tA} \fsto \ti{\tS}\;\; \mbox{ i.e.}\;\;
\forall N \in \ti{\tA}. \; \tei{\lambda x.M}_{\rho}N \in \ti{\tS}.$$
Suppose that $N \in \ti{\tA}$.
We have that $\rho(N/x) \models \Gamma, x:\alpha$ (**) since $\rho \models \Gamma$, $x \not\in \Gamma$ and $\rho(N/x)(x)=N \in \ti{\tA}$. From (*) and (**) we conclude that $\rho(N/x) \models M:\sigma$, hence we can conclude that $\tei{M}_{\rho(N/x)} \in \ti{\tS}$. Using Lemma~\ref{lemma:val}(iii) we get $\tei{\lambda x.M}_{\rho} N \to_{\beta} \tei{M}_{\rho}\ISUB{N}{x} = \tei{M}_{\rho(N/x)}$. Since $\tei{M}_{\rho(N/x)} \in \ti{\tS}$ and $\ti{\tS}$ is $\circledR$-saturated,
we obtain $\tei{\lambda x.M}_{\rho} N \in \ti{\tS}$. 





\item
The last rule applied is $(\to_{E})$, i.e.\
$$\infer[(\to_E)]{\Gamma, \dztop \sqcap \Delta_1 \sqcap ...
\sqcap \Delta_n \vdash MN:\tS}
                    {\Gamma \vdash M:\cap^n_i \tT_i \to \tS & \Delta_0 \vdash N:\tT_0\; \ldots\; \Delta_n \vdash N:\tT_n}$$
Let $`r$ be any valuation.
Assuming that
$`G "|-" M: \cap_{i}^{n} \tT_i \to `s, `D_0 "|-"  N:\tT_0,\ldots, `D_n"|-" N:\tT_n$, we have to prove that if
$`r "|=" `G, \dztop \sqcap `D_1  \sqcap ... \sqcap `D_n$, then $`r "|=" M\,N:`s$, i.e.\ $\tei{MN}_{\rho} \in \ti{\sigma}.$

By IH, $\Gamma \models M :\cap_i^n \tT_i \to \tS$ and $\Delta_0 \models N:\tT_0, \ldots, \Delta_n \models N:\tT_n$.
Assume that $\rho \models \Gamma, \dztop \sqcap \Delta_1 \sqcap \ldots \sqcap \Delta_{n}$. This means that $\rho \models \Gamma$ and $\rho \models \dztop \sqcap \Delta_1 \sqcap \ldots \sqcap \Delta_{n}.$
From $\rho \models \Gamma$ we deduce by Definition~\ref{def:model}~(iii) $\rho \models M: \cap_{i}^{n} \tT_i \to \tS$ and by Definition~\ref{def:model}~(i) $\tei{M}_{\rho} \in \ti{\cap_{i}^{n}
  \tT_{i} \to \tS}$. By Definition~\ref{def:val} $\tei{M}_{\rho} \in \bigcap_{i}^n
\tei{\tT_{i}} \fsto \tei{\tS}$ (*).
Using Lemma~\ref{lem:valcup} $\rho \models \dztop \sqcap \Delta_1 \sqcap ... \sqcap \Delta_n$ implies
$(\rho \models \dztop) \wedge (\bigwedge_{i=1}^n\rho \models \Delta_{i})$, hence by
Definition~\ref{def:model}~(i) and (iii) we get $(\tei{N}_{\rho} \in \ti{\top}) \wedge
{\bigwedge_{i=1}^n(\tei{N}_{\rho} \in \ti{\tT_{i}})}$, i.e.\ $\tei{N}_{\rho} \in \SN \ \cap\ \cap_{i}^{n }\ti{\tT_{i}} = \cap_{i}^{n }\ti{\tT_{i}}$ (**), since $\ti{\tau_i} \subseteq \SN$ by Proposition~\ref{prop:saturSets}(iv). From (*) and (**), using Definition~\ref{def:fsto} of $\fsto\!\!\!\!$, we can conclude that $\tei{M}_{\rho} \tei{N}_{\rho} \in \tei{\tS}$. Using Lemma~\ref{lemma:val}(ii) we can conclude that $\tei{M\,N}_{\rho} = \tei{M}_{\rho} \tei{N}_{\rho}
\in \tei{\tS}$  and by Definition~\ref{def:model}~(i) ${`r "|=" M\,N:`s}$.

\item
The last rule applied is $(Thin)$, i.e.,\
$$\infer[(Thin)]{\Gamma, x:\top \vdash \weak{x}{M}:\tS}
                    {\Gamma \vdash M:\tS}$$
By the IH $\Gamma \models M:\tS$.
Suppose that $\rho \models \Gamma, x:\top$ $\Leftrightarrow$  $\rho \models \Gamma$ and $\rho \models x:\top$. From $\rho \models \Gamma$ we obtain $\tei{M}_{\rho} \in \ti{\tS}$. Using multiple times the thinning property $\WEAK(\ti{\sigma})$ and Lemma~\ref{lemma:val}(iv) we obtain $\weak{Fv(\rho(x))}{\tei{M}_{\rho}} = \tei{\weak{x}{M}}_{\rho} \in \ti{\tS}$, since $Fv(\rho(x)) \cap Fv(\tei{M}_{\rho}) = \emptyset$.

\item
The last rule applied is $(Cont)$, i.e.,\
$$\infer[(Cont)]{\Gamma, z:\tA \cap \tB \vdash
\cont{z}{x}{y}{M}:\tS}
                    {\Gamma, x:\tA, y:\tB \vdash M:\tS}$$
By the IH $\Gamma, x:\tA, y:\tB \models M:\tS$.
Suppose that $\rho \models \Gamma, z:\tA \cap \tB$.
This means that $\rho \models \Gamma$ and $\rho \models z:\tA \cap \tB$ $\Leftrightarrow$ $\rho(z) \in \ti{\tA} \mbox{ and } \rho(z) \in \ti{\tB}$.
For the sake of simplicity let $\rho(z) \equiv N$. We define a new valuation $\rho'$ such that $\rho' = \rho (\isubs{N_{1}}{x}, \isubs{N_{2}}{y})$, where $N_{1}$ and $N_{2}$ are obtained by renaming the free variables of $N$.
Then $\rho' \models \Gamma, x:\tA, y:\tB$ since $x,y \not \in Dom(\Gamma)$, $N_{1} \in \ti{\tA}$ and $N_{2} \in \ti{\tB}$.
By the IH $\tei{M}_{\rho'} = \tei{M}_{\rho(\isubs{N_{1}}{x}, \isubs{N_{2}}{y})}  \in \ti{\tS}$.
Using the contraction property $\CONT(\ti{\sigma})$ and Lemma~\ref{lemma:val}(v) we have that
$\cont{Fv(N)}{Fv(N_{1})}{Fv(N_{2})}{\tei{M}_{\rho (\isubs{N_{1}}{x}, \isubs{N_{2}}{y})}} =
\tei{\cont{z}{x}{y}{M}}_{\rho} \in \ti{\tS}$.

\end{itemize}
\end{proof}

\begin{theorem} [$\SN$ for $\rcl \cap$] \label{th:typ=>SN}
If $\Gamma \vdash M:\tS$, then $M$ is strongly normalising, i.e. $M \in \SN$.
\end{theorem}

\begin{proof}
Suppose $\Gamma \vdash M:\tS$. By Proposition~\ref{prop:sound}\; $\Gamma \models M:\tS$. According to Definition~\ref{def:model}(iii), this means that $ (\forall\rho)\; \rho \models \Gamma \quad "=>"\quad \rho \models M : \tS$. We can choose a particular $\rho_{0}(x) = x$ for all $x \in {\tt var}$. By Proposition~\ref{prop:saturSets}(iv), $\ti{\tS}$ is $\circledR$-saturated for each type $\tS$, hence $\tei{x}_{\rho_0} = x  \in \ti{\tS}$ (variable condition for $n=0$). Therefore, $\rho_{0} \models \Gamma$ and we can conclude that $\tei{M}_{\rho_{0}} \in \ti{\tS}$. On the other hand, $M = \tei{M}_{\rho_{0}}$ and $\ti{\tS} \subseteq \SN$ (Proposition~\ref{prop:saturSets}), hence $M \in \SN$.
\end{proof}

Finally, we can give a characterisation of strong normalisation in $\rcl$-calculus.

\begin{theorem}
In $\rcl$-calculus, the term $M$ is strongly normalising if and only if it is typeable in $\rcl\cap$.
\end{theorem}

\begin{proof}
Immediate consequence of Theorems~\ref{th:typ=>SN} and~\ref{thm:SNtypable}.
\end{proof}


\section{Related work and conclusions}
\label{sec:conclusion}

The idea to control the use of variables can be traced  back to Church's $\lambda I$-calculus~\cite{bare84} and Klop's extension of $\lambda$-calculus~\cite{klop80}.
Currently, there are several different lines of research in resource aware term calculi.

Van Oostrom~\cite{oost01} and later Kesner and Lengrand~\cite{kesnleng07}, applying ideas from linear logic~\cite{LL}, proposed to extend $\lambda$-calculus with explicit substitution~\cite{kesnleng07} with operators to control the use of variables (resources). Their linear $\llxr$-calculus is an extension of the ${\lambda \mathsf{x}}$-calculus~\cite{BlooRose95,RoseBlooLang:jar2011}
with operators for linear substitution, erasure and duplication, preserving at the same time confluence
and full composition of explicit substitutions. 
The
simply typed version of this calculus corresponds to the
intuitionistic fragment of linear logic proof-nets, according to
Curry-Howard correspondence, and it enjoys strong normalisation
and subject reduction. Generalising this approach, Kesner and Renaud~\cite{kesnrena09,kesnrena11}
developed the {\em prismoid of resources}, a system of eight calculi parametric over the explicit
 and implicit treatment of substitution, erasure and duplication.   

On the other hand, process calculi and their relation to $\lambda$-calculus by Boudol~\cite{boud93} initialised
investigations in resource aware non-deterministic $\lambda$-calculus with multiplicities and a generalised notion of application~\cite{boudcurilava99}. The theory was connected to linear logic via differential
$\lambda$-calculus by Ehrhard and Regnier in~\cite{ehrhregn03} and typed with non-idempotent intersection types by Pagani and Ronchi Della Rocha in~\cite{pagaronc10}. An account of this approach is given in~\cite{alvefernflorimack14}.

Resource control in sequent
calculus corresponding to classical logic was proposed by \v Zuni\'c in~\cite{zunicPHD}. Resource control in sequent $\lambda$-calculus was investigated in~\cite{ghilivetlesczuni11}.

Intersection types in the presence of resource control were first introduced in ~\cite{ghilivetlikalesc11}. Later on non-idempotent intersection types for $\llxr$-calculus were introduced by Bernadet and Lengrand in~\cite{bernleng13}. Their proof of strong normalisation takes advantage of intersection types being non-idempotent. 

Our contribution extends the work of \cite{ghilivetlikalesc11}, accordingly
we follow the notation of \cite{zunicPHD} and \cite{ghilivetlikalesc11}, along the lines of~\cite{oost01}. We have proposed an intersection type assignment system for the
resource control lambda calculus $\rcl$,
which gives
a complete characterisation of strongly normalising terms of the
$\rcl$-calculus. The  proofs do not rely on any assumption about idempotence, hence they can be applied both to idempotent and non-idempotent intersection types.

This paper expands the range of
the intersection type techniques and combines different methods in
the strict type environment. It should be noticed that the
strict control on the way variables are introduced determines the
way terms are typed in a given environment. Basically, in a given
environment no irrelevant intersection types are introduced. The
flexibility on the choice of a type for a term, as it is used in
rule $(\to_E)$ in Figure~\ref{fig:typ-rcl-int}, comes essentially
from the choice one has in invoking the axiom. 

The presented calculus is a good
candidate to investigate the computational content of
substructural logics~\cite{schrdose93} in natural deduction style and relation to substructural type systems~\cite{walk05}.
The motivation for these logics comes from
philosophy (Relevant Logics), linguistics (Lambek Calculus), 
computing (Linear Logic). Since the basic idea of resource control
is to explicitly handle structural rules, the control operators
could be used to handle the absence of (some) structural rules in
substructural logics such as thinning, weakening, contraction,
commutativity, associativity. This would be an interesting
direction for further research.
Another direction involves the
investigation of the use of intersection types, being a powerful
means for building models of lambda
calculus~\cite{barecoppdeza83,dezaghillika04}, in constructing
models for substructural type systems.
Finally, one may wonder how the
strict control on the duplication and the erasure of variables
influences the type reconstruction of
terms~\cite{DBLP:journals/entcs/BoudolZ05,kfouwell04}.

\textbf{Acknowledgements:} We would like to thank
anonymous referees of a previous version of this paper for their careful reading and many valuable
comments, which helped us to improve the paper.
We would also like to thank Dragi\v sa \v Zuni\' c for
participating in the earlier stages of the work.  This work is partially supported by the Serbian Ministry of Science - project ON174026 and by a bilateral project between Serbia and France within the ``Pavle Savi\'c" framework. 


\end{document}